\RequirePackage{xcolor}
\documentclass[journal,twoside,web]{ieeecolor}

\usepackage[english]{babel}
\usepackage{booktabs}
\usepackage{siunitx}
\usepackage{xcolor}
\usepackage{nicefrac}
\usepackage{comment}
\usepackage{tikz}
\usepackage{soul}
\usepackage{generic}
\usepackage{cite}
\usepackage{amsmath,amssymb,amsfonts}
\usepackage{algorithmic}
\usepackage{graphicx}
\usepackage{algorithm,algorithmic}
\usepackage{hyperref}
\hypersetup{hidelinks=true}
\usepackage{textcomp}
\usepackage{float}
\usetikzlibrary{arrows.meta,positioning,calc}

\usepackage{balance}

\newtheorem{theorem}{Theorem}
\newtheorem{corollary}{Corollary}
\newtheorem{lemma}{Lemma}
\newtheorem{proposition}{Proposition}

\newcommand{\circledlabel}[2]{%
  \tikz[baseline=(n.base)]\node[circnum] (n) {#1};\ #2%
}

\title{Lotka-Sharpe Neural Operators \\ for Control of Population  PDEs}
\author{Miroslav Krsti\'c, Iasson Karafyllis, Luke Bhan, Carina Veil
\thanks{Miroslav Krsti\'c and Luke Bhan are with the University of California, San Diego, USA, \{mkrstic, lbhan\}@ucsd.edu. Iasson Karafyllis is with the National Technical University of Athens, Greece, iasonkar@central.ntua.gr. Carina Veil is with KTH Royal Institute Technology, Stockholm, Sweden, veil@kth.se. 
}
}

\begin{document}
\maketitle
\begin{abstract}
Age-structured predator-prey integro-partial differential equations provide models of  interacting populations in ecology, epidemiology, and biotechnology. A key challenge in feedback design for these systems is the scalar $\zeta$, defined implicitly by the Lotka-Sharpe nonlinear integral condition, as a  mapping from fertility and mortality rates to $\zeta$. To solve this challenge with operator learning, we first prove that the Lotka-Sharpe operator is Lipschitz continuous, guaranteeing the existence of arbitrarily accurate neural operator approximations over a compact set of fertility and mortality functions. We then show that the resulting approximate feedback law preserves semi-global practical asymptotic stability under propagation of the operator approximation error through various other nonlinear operators, all the way through to the control input. In the numerical results, not only do we learn ``once-and-for-all'' the canonical Lotka-Sharpe (LS) operator, and thus make it available for future uses in control of other age-structured population interconnections, but we demonstrate the online usage of the neural LS operator under estimation of the fertility and mortality functions. 

\end{abstract}

\section{Introduction}

Understanding the dynamics of interacting populations is fundamental to predicting the behavior of ecosystems, epidemics, and bioreactors. Age-structured population models, formulated as partial differential equations (PDEs), have emerged as the cornerstone framework for describing such dynamics \cite{mckendrick1925Applications}. While single-species dynamics have received considerable attention from a control-theoretic perspective \cite{inaba2017age, iannelli2017basic, dochain2013automatic, schmidt2018yield, kurth2021tracking, kurth2023control, karafyllis2017stability, haacker2024stabilization, karafyllis2025agestructuredchemostatsubstratedynamics, BARGO2026104629}, the multi-species setting — where predator and prey populations are coupled through nonlinear feedback — has only recently been tackled with a rigorous feedback design \cite{veil2025stabilization, veil2025stabilizationagestructuredcompetingpopulations, veil2025stabilizationExtendedPreprint}. 
These results establish global stabilization laws capable of driving both populations to prescribed set-points. However, they share a common vulnerability: each controller gain depends critically on a scalar $\zeta$, defined implicitly through the Lotka--Sharpe (LS) condition \cite{sharpe1911problem}. In general, this scalar cannot be computed in closed form and must therefore be approximated numerically.

Biologically, $\zeta$ is critical as it encodes the long-term fate of the population: whether it grows, declines, or reaches equilibrium. Furthermore, from a mathematical perspective, it is equally rich, representing an infinite dimensional mapping from the birth rate $k(a)$ and mortality rate $\mu(a)$ into a real-valued scalar $\zeta$ satisfying:
\begin{equation}
\label{eq:Lotkasharpe}
    \int_0^A k(a) e^{-\int_0^a(\mu(s) + \zeta)ds} da = 1\,, 
\end{equation}
 Generally, this infinite dimensional map has no analytical solution and hence every change in $k$ or $\mu$ demands a fresh computation. In age dependent control designs, the Lotka-Sharpe parameters $\zeta_1$ and $\zeta_2$ of both the predator and prey species appear in the controller gains. Consequently, any practical implementation will introduce approximation errors with no a priori stability certificate. This paper is the first step toward understanding approximations of this operator and certifying stability of feedback laws when faced with such approximations.

To begin studying approximations of the Lotka-Sharpe operator, we first establish the mapping $(k, \mu) \mapsto \zeta$ is Lipschitz continuous. This is the key technical challenge as $\zeta$ is only given implicitly and hence its continuity requires a monotonicity argument tailored to the biological constraints of the domain. Establishing this continuity paves the way for both neural operator \cite{chen1995Universal, lu2021Learning, li2021Fourier} and numerical approximations. In this work, we focus on the operator learning paradigm as it has proven transformative for replacing expensive implicit computations in feedback laws, with deployments spanning PDE backstepping \cite{10374221, krstic2024Neural, wang2024Backstepping}, adaptive control \cite{lamarque2025Adaptive, lamarque2025Gain, BHAN2025105968}, delayed systems \cite{wang2024Backstepping}, and applications in biological Chemostats \cite{bhan2025stabilizationnonlinearsystemsunknown} as well as traffic flows \cite{pmlr-v242-zhang24c, LYU202513}. Building on the Lipschitz continuity of the Lotka--Sharpe operator, we establish a universal approximation theorem over compact classes of birth and mortality profiles.

Moreover, we do not stop at just the approximation. For the predator--prey model, we study the robustness of the feedback law when the exact Lotka--Sharpe parameters are replaced by approximations. In particular, we prove semi-global practical asymptotic stability of the resulting closed-loop system. We emphasize that this stability result is \textbf{not} limited to neural operators, but is a robustness result that captures any uniform approximation. In this sense, the paper resolves a foundational vulnerability shared by every existing age-structured predator-prey controller: for the first time, one can implement the feedback law without exact knowledge of $\zeta$ and still have a rigorous guarantee that the populations will behave.

The paper makes the following specific contributions, through new ideas, techniques, and results:
\begin{enumerate}
\item \emph{Formulation of the Lotka–Sharpe mapping as an operator in feedback control.} The dependence of stabilizing predator–prey controllers on the implicitly defined scalar $\zeta$ is recast as an operator mapping from functional data (birth and mortality profiles) to a control-relevant parameter, thereby exposing a heretofore hidden concept in age-structured feedback design.
\item \emph{Establishment of Lipschitz continuity for an implicitly defined nonlinear operator on a biologically constrained domain.} A nonstandard analysis is developed to prove Lipschitz continuity of the Lotka–Sharpe operator despite its implicit definition, leveraging monotonicity properties induced by ordered bounds on fertility and mortality functions.
\item \emph{Derivation of a universal approximation framework for the Lotka–Sharpe operator.} By combining the Lipschitz property with compactness of admissible function classes, it is shown that the Lotka–Sharpe operator admits uniformly accurate neural operator approximations, placing it within the operator-learning paradigm despite its implicit structure.
\item \emph{Explicit characterization of how approximation errors propagate through the control architecture.} It is clearly exhibited how approximation errors in $\zeta$ do not remain localized but enter all downstream operator evaluations ($\mathcal{G}_\kappa$, $\mathcal{G}_\gamma$, $\mathcal{G}_\pi$), yielding a structured perturbation of the control law that is reduced to two scalar error channels.
\item \emph{Development of a robustness analysis for controllers with implicitly parameterized gains under approximation.} A new analytical framework is constructed to handle perturbations that simultaneously affect multiple gain terms through a shared implicit parameter. The technical approach is devised to yield guarantees in the presence of a positivity constraint on the control input and under Lyapunov derivatives that,  while negative definite, are non-proper, because, in the context of population dynamics,  extinction is a barrier near the equilibrium.
\end{enumerate}

The paper is organized as follows.  Section \ref{sec-pop} presents the age-structured predator–prey model and the Lotka–Sharpe condition. Section \ref{sec-oper} introduces the four operators and highlights the implicit nature of the Lotka–Sharpe operator. Section \ref{sec-controllers} develops the nominal and approximate control laws and identifies the structure of the induced perturbation. Section \ref{sec-approx} establishes neural approximability of the Lotka–Sharpe operator. Section \ref{sec-stabilization} provides the stability analysis under approximation errors, including the main robustness result. Section \ref{sec-proofs} contains the proofs of the main theorems. Section \ref{sec-sims} presents numerical results and Section \ref{sec-adaptive} presents a illustration of a adaptive design when the fertility and mortality functions are unknown. 

\color{blue}
A preliminary version of this paper has been submitted to the Conference on Decision and Control 2026 \cite{cdc2026version}. This journal version, additionally, contains all the proofs (Section \ref{sec-proofs} and Appendices \ref{app-functions+constants}, \ref{app-Lip-other}, \ref{app-bound}) as well as an additional illustrative adaptive design in Section \ref{sec-adaptive}.
\color{black}

\textbf{Notation:}
Denote the sets $\mathbb{R}_{>0}$ and $\mathbb{R}_{\geq 0}$ as the positive real numbers excluding and including zero respectively. Let $C^k(S_1;S_2)$ represent the class of $k \geq 1$ continuously differentiable functions mapping $S_1$ to $S_2$ and $C^0(S_1; S_2)$ be the class of continuous functions mapping $S_1$ to $S_2$. Let $A> 0$ be a real-valued positive scalar. For a function $f: [0, A] \to \mathbb{R}_{\geq 0}$, we define $\|f(\cdot)\|_\infty$ to be the supremum norm $\sup_{x \in [0, A]}|f(x)|$. For a distributed function $f(a, t)$ with $(a, t) \in [0, A] \times \mathbb{R}_{\geq 0}$, we use $\dot{f} = \frac{\partial f}{\partial t}$ for the time derivative and analogously $f' = \frac{\partial f}{\partial x}$ for the space derivative.  

\section{Age-structured population model}
\label{sec-pop}

The dynamics of one age-structured species in a chemostat with population density $x(a,t)$, where the organisms compete for a common food source, is governed by
\begin{align}
    x'(a,t) + \dot x(a,t) =\;& - x(a,t) \Big[ \mu(a)\nonumber \\ & + \int_0^A p(\alpha)x(\alpha,t) d\alpha 
    +u(t) \Big] \label{eq:one-pop-pde}
\end{align}
with mortality function $\mu(a)$, competition kernel $p(a)$, dilution $u(t)$, and derivatives with respect to time and $x'$ with respect to age \cite{mckendrick1925Applications}.
Essentially, the population density is reduced by mortality, competition, and dilution.

The \textbf{Lotka-Sharpe condition} \eqref{eq:Lotkasharpe}
 is an integral equation that defines the intrinsic growth rate $\zeta$. 
It ensures that the mortality-discounted age-specific fertility contributions equal one, which characterizes a stable population growth rate and age distribution.
It was proven in \cite{sharpe1911problem}  that  \eqref{eq:Lotkasharpe}
has a unique positive real-valued solution $\zeta(k,\mu)$ for any nonnegative measurable birth rate function $k$ that is not identically zero and for any nonnegative measurable mortality rate function $\mu$, such that $\int_0^A k(a)e^{-\int_0^a \mu(s)ds}da > 1$.

Extending \eqref{eq:one-pop-pde} to a predator-prey setup results in the following age-structured model considered in \cite{veil2025stabilization}, with initial conditions (IC) and boundary conditions (BC),
\begin{subequations}
	\begin{alignat}{3}
	&\frac{\partial x_1}{\partial t}(a,t) + \frac{\partial x_1}{\partial a}(a,t) &&=  
	- x_1(a,t)  \Bigg[ \mu_1(a) + u(t)  
  \nonumber \\ & && \qquad  + \int_0^A g_1(\alpha) x_2(\alpha,t) d \alpha 
    \Bigg]  \label{eq:ipde1}\\ 
 	&\frac{\partial x_2}{\partial t}(a,t) + \frac{\partial x_2}{\partial a}(a,t) &&= 
	- x_2(a,t)  \Bigg[ \mu_2(a) + u(t) 
      \nonumber \\ & && \qquad + \frac{1}{\int_0^A g_2(\alpha) x_1(\alpha,t) d \alpha} 
    \Bigg] 
     \label{eq:ipde2}\\ 
	&\text{IC}: \qquad\quad \ \ x_i(a, 0) &&=\  x_{i,0}(a), \label{eq:ic} \\
	&\text{BC}: \qquad\quad \, \ x_i(0,t) &&=\ \int_0^A k_i(a) x_i(a,t) d a  \label{eq:bc},
 \end{alignat} \label{eq:pred-prey-model}%
\end{subequations}
where, for $i,j \in \{1,2\}$, $i\neq j$, $x_i(a,t)>0$ is the population density, i.~e. the amount of organisms of a certain age $a \in [0,A]$ of the two interacting populations $x_1(a,t)$ and $x_2(a,t)$ with $(a,t) \in [0,A] \times \mathbb{R}_{> 0}$, their derivatives $\dot{x}_i$ with respect to time and $x'_i$ with respect to age, and the constant maximum age $A>0$. 
The interaction kernels $g_i(a):[0,A]\rightarrow\mathbb{R}_{\geq 0}$, the mortality rates $\mu_i(a):[0,A]\rightarrow\mathbb{R}_{\geq 0}$, and the birth rates $k_i(a):[0,A]\rightarrow\mathbb{R}_{\geq 0}$ are continuous functions with $\int_0^A \mu_i(a)d a > 0$, $\int_0^A g_i(a)d a > 0$, $\int_0^A k_i(a)d a >0$. The continuous dilution rate $u(t):\mathbb{R}_{\geq 0}\rightarrow \mathbb{R}_{\geq 0}$, is an input affecting both species.

\begin{proposition}[Equilibrium \cite{veil2025stabilization}]
\label{prop:steady}
The equilibrium  state $(x_1^*(a),x_2^*(a))$ of the population system (\ref{eq:pred-prey-model}), along with the equilibrium dilution input $u^*$, is given by
\begin{subequations}
\begin{align}
x_i^* (a) &= x_i^*(0)\, n_i(a)\,, \qquad n_i(a):= {e^{-{\int_0^a (\zeta_i + \mu_i(s)) d s}}},\label{eq:ss_profiles}\\
u^* &= \zeta_1 -  \lambda_2 = \zeta_2 - \frac{1}{\lambda_1}  \in \left(0, \min\left\{\zeta_1 , \zeta_2\right\}\right)\,,
\label{eq:u_star_constraint}
\end{align}
\end{subequations}
with unique parameters $\zeta_i(k_i,\mu_i)$ resulting from the Lotka-Sharpe condition \cite{sharpe1911problem}, 
\begin{subequations}
    \label{eq:def_lambda}
\begin{eqnarray}
\lambda_1&:=& \int_0^A g_2(a) x_1^*(a) d a = x_1^*(0)\gamma_2\,,   \\   
\gamma_2 &:=& \int_0^A g_2(a)n_1(a) da >0 \\
\lambda_2&:=& \int_0^A g_1(a) x_2^*(a) d a = x_2^*(0) \gamma_1 \\ 
\gamma_1 &:=& \int_0^A g_1(a)n_2(a) da>0
\end{eqnarray}
\end{subequations}
and the positive  concentrations of the newborns 
\begin{subequations}
    \begin{align}
    x_1^*(0) =&\;  \frac{1}{\left(\zeta_2-u^*\right) 
    \gamma_2} >0, \\
    x_2^*(0) =&\; \frac{\zeta_1-u^*}{
    \gamma_1}  = \frac{1}{\gamma_1} \left[\zeta_1 - \zeta_2
        + \frac{1}{x_1^*(0)\gamma_2} \right]> 0\,. \label{eq-x2star-zero}
    \end{align}
    \label{eq:ic_constraint}
\end{subequations}
Moreover, for $u^*$ to be positive, the prey birth concentration must be commanded to be large enough: 
\begin{align}
    x_1^*(0) > \frac{1}{\zeta_2
    \gamma_2 }.
\end{align}
\end{proposition}

Proposition \ref{prop:steady} indicates that the equilibrium is explicitly characterized by the Lotka--Sharpe quantities $\zeta_i$, the dilution setpoint $u^\ast$, and the newborn concentrations $x_i^\ast(0)$. To achieve a target equilibrium, we will invoke the feedback law designed in \cite{veil2025stabilization}. However, before doing so, we first define four key operators necessary for implementing the feedback law.

\begin{figure}[t]
\centering
\resizebox{0.5\textwidth}{!}{%
\begin{tikzpicture}[
    x=1pt,y=1pt,
    >=Latex,
    line/.style={draw=black!55, line width=1.4pt},
    thinline/.style={draw=black!45, line width=0.9pt, dash pattern=on 3pt off 2pt},
    box/.style={
        rounded corners=8pt,
        draw=black!45,
        line width=0.7pt,
        fill=black!5
    },
    greenbox/.style={
        rounded corners=8pt,
        draw=blue!65!black,
        line width=0.8pt,
        fill=blue!10
    },
    purplebox/.style={
        rounded corners=8pt,
        draw=orange!85!black,
        line width=0.8pt,
        fill=orange!12
    },
    label/.style={font=\Large, text=black!75},
    title/.style={font=\LARGE, text=black!100},
    bigtitle/.style={font=\LARGE\bfseries, text=black!88},
    subtitle/.style={font=\Large, text=black!70},
    note/.style={font=\Large, text=black!82},
    gnote/.style={font=\Large, text=blue!65!black},
    pnote/.style={font=\Large, text=orange!85!black}
]

\def\boxh{46}
\def\ysep{30}
\def\yTopA{-126}

\def\xAL{170} \def\xAR{330}
\def\xBL{70}  \def\xBR{275}
\def\xCL{70}  \def\xCR{275}
\def\xDL{70}  \def\xDR{275}

\pgfmathsetmacro{\xCtrA}{(\xAL+\xAR)/2}
\pgfmathsetmacro{\xCtrB}{(\xBL+\xBR)/2}
\pgfmathsetmacro{\xCtrC}{(\xCL+\xCR)/2}
\pgfmathsetmacro{\xCtrD}{(\xDL+\xDR)/2}

\pgfmathsetmacro{\yBotA}{\yTopA-\boxh}
\pgfmathsetmacro{\yTopB}{\yBotA-\ysep}
\pgfmathsetmacro{\yBotB}{\yTopB-\boxh}
\pgfmathsetmacro{\yTopC}{\yBotB-\ysep}
\pgfmathsetmacro{\yBotC}{\yTopC-\boxh}
\pgfmathsetmacro{\yTopD}{\yBotC-\ysep}
\pgfmathsetmacro{\yBotD}{\yTopD-\boxh}

\pgfmathsetmacro{\yCtrA}{(\yTopA+\yBotA)/2}
\pgfmathsetmacro{\yCtrB}{(\yTopB+\yBotB)/2}
\pgfmathsetmacro{\yCtrC}{(\yTopC+\yBotC)/2}
\pgfmathsetmacro{\yCtrD}{(\yTopD+\yBotD)/2}

\pgfmathsetmacro{\yAnnA}{\yTopA + \ysep + 6}
\pgfmathsetmacro{\yAnnB}{\yCtrA}
\pgfmathsetmacro{\yAnnC}{\yCtrB}
\pgfmathsetmacro{\yAnnD}{\yCtrC}

\pgfmathsetmacro{\yOut}{\yBotD-22}
\pgfmathsetmacro{\yTitle}{\yOut-28}
\pgfmathsetmacro{\yFrameBot}{\yTitle+24}

\draw[draw=black!20, line width=0.8pt, dash pattern=on 6pt off 3pt, rounded corners=14pt]
  (50,-20) rectangle (520,\yFrameBot-18);

\draw[line,->] (100,-50) -- (100,-118);
\node[label] at (100,-44) {$k(a)$};

\draw[line,->] (250,-50) -- (250,\yTopA+6);
\node[label] at (250,-44) {$\mu(a)$};

\node[note, anchor=west, align=left, text width=100pt] at (130,-70)
  {extend by $0$ \\ for $a > A$};

\draw[box, fill=black!7] (\xAL,\yTopA) rectangle (\xAR,\yBotA);
\node[title] at (\xCtrA,\yCtrA+10) {Survival $\Pi(a)$};
\node[subtitle] at (\xCtrA,\yCtrA-10) {$e^{-\int_0^a \mu(s)\,ds}$};

\draw[line] (100,-118) -- (100,\yTopB+6);
\draw[line,->] (250,\yBotA) -- (250,\yTopB+2);
\node[note, anchor=west] at (254,{(\yBotA+\yTopB)/2}) {$\Pi(a)$};

\draw[box, fill=black!7] (\xBL,\yTopB) rectangle (\xBR,\yBotB);
\node[title] at (\xCtrB,\yCtrB+10) {Net maternity $k\!\cdot\!\Pi$};
\node[subtitle] at (\xCtrB,\yCtrB-10) {$k(a)\cdot \Pi(a)$};

\draw[line,->] (100,\yTopB+6) -- (100,\yTopB);
\draw[line,->] (250,\yTopB+2) -- (250,\yTopB);

\draw[line,->] (\xCtrB,\yBotB) -- (\xCtrC,\yTopC+6);
\node[note, anchor=west] at (183,{(\yBotB+\yTopC+6)/2}) {$k(a)\cdot \Pi(a)$};

\draw[greenbox] (\xCL,\yTopC) rectangle (\xCR,\yBotC);
\node[title, text=blue!65!black] at (\xCtrC,\yCtrC+10) {Laplace transform};
\node[subtitle, text=blue!60!black] at (\xCtrC,\yCtrC-10)
  {$\mathcal{L}\{k\Pi\}(\zeta)=\int_0^\infty k(a)\Pi(a)e^{-\zeta a}\,da$};

\draw[line,->] (\xCtrC,\yBotC) -- (\xCtrD,\yTopD+6);
\node[gnote, anchor=west] at (183,{(\yBotC+\yTopD+6)/2-3})
  {$F(\zeta)=\mathcal{L}\{k\Pi\}(\zeta)$};

\draw[purplebox] (\xDL,\yTopD) rectangle (\xDR,\yBotD);
\node[title, text=orange!85!black] at (\xCtrD,\yCtrD+10) {Function inversion $\star$};
\node[subtitle, text=orange!80!black] at (\xCtrD,\yCtrD-10) {find $\zeta$ : $F(\zeta)=1$};

\draw[draw=orange!85!black, line width=1.2pt, ->] (\xCtrD,\yBotD) -- (\xCtrD,\yOut);
\node[pnote, anchor=west] at (183,\yOut+6) {$\zeta$};

\node[note, anchor=west, align=left, text width=150pt] at (365,\yAnnA)
  {step 1 --- trivial:\\ \quad (zero-extend $k$)};
\draw[thinline] (150,\yAnnA) -- (365,\yAnnA);

\node[note, anchor=west, align=left, text width=150pt] at (365,\yAnnB)
  {step 2 --- direct:\\ \quad (explicit integral)};
\draw[thinline] (365,\yAnnB) -- (335,\yAnnB);

\node[note, anchor=west, align=left, text width=150pt] at (365,\yAnnC)
  {step 3 --- direct:\\ \quad (pointwise product)};
\draw[thinline] (365,\yAnnC) -- (\xBR+2,\yAnnC);

\node[gnote, anchor=west, align=left, text width=150pt] at (365,\yAnnD)
  {step 4 --- direct:\\ \quad (linear integral)};
\draw[draw=blue!55!black, line width=0.5pt, dash pattern=on 3pt off 2pt]
  (365,\yAnnD) -- (\xCR+2,\yAnnD);

\node[pnote, anchor=west, align=left, text width=150pt] at (365,\yCtrD)
  {step 5 --- nontrivial $\star$\\ (implicit, no closed form)};
\draw[draw=orange!80!black, line width=0.5pt, dash pattern=on 3pt off 2pt]
  (365,\yCtrD) -- (\xDR+2,\yCtrD);

\node[bigtitle, align=center] at (300,\yTitle+30) {$G_{LS} : (k,\mu)\mapsto \zeta$};
\end{tikzpicture}%
}
\caption{Computational breakdown of the Lotka-Sharpe operator}
\label{figure:Lotka-sharpe}
\end{figure}
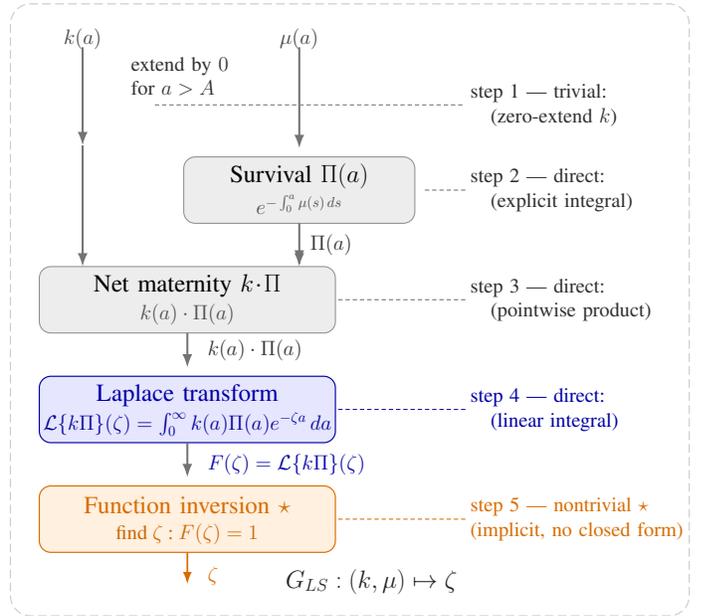

\section{Four Operators}
\label{sec-oper}

In the implementation of a stabilizing controller for an age-structured predator-prey system, four operators are involved, the principal among which is the Lotka-Sharpe operator $\mathcal{G}_{\rm LS}$ (See Figure \ref{fig:implicit-explicit-operators}). 

\subsection{Lotka-Sharpe operator (output sits {\em within} an integral --- it has to be solved for).}
Define the Lotka-Sharpe operator as $\mathcal{G}_{\rm LS}$ : $(k, \mu) \mapsto \zeta$, mapping two functions into a scalar, and defined implicitly by
\begin{equation}\int_0^A k(a)\, e^{-\int_0^a (\zeta+\mu(s))\,ds}\,da = 1\end{equation}
with $\zeta_i = \mathcal{G}_{\rm LS}(k_i, \mu_i)$, $i=1,2$.

For $A=\infty$, the Lotka-Sharpe condition admits a useful reformulation that makes the mathematical meaning of $\mathcal{G}_{\rm LS}$ transparent. Define the survival function
\begin{align}
\Pi(a)=\exp\!\Big(-\int_0^a \mu(s)\,ds\Big),
\end{align}
so that $k(a)\Pi(a)$ is the {\em net maternity function}, i.e.\ fertility at age $a$ weighted by survival up to age $a$. Then the defining equation becomes
\begin{align}
\int_0^\infty k(a)\Pi(a)e^{-\zeta a}\,da = 1,
\end{align}
or equivalently
\begin{align}
\mathcal{L}\{k\Pi\}(\zeta)=1.
\end{align}
where $\mathcal{L}$ is the Laplace transform. 
Hence, the Lotka-Sharpe operator may be viewed as
\begin{align}
\mathcal{G}_{\rm LS}(k,\mu)=(\mathcal{L}\{k\Pi\})^{-1}(1),
\end{align}
namely: first form the net maternity profile $k\Pi$, then take its Laplace transform, and finally locate its $1$-level crossing. Biologically, $\zeta$ is the harvesting rate for which discounted lifetime maternity is exactly one, so that each individual replaces itself and the population remains constant.


This decomposition also clarifies what is mathematically easy and what is genuinely difficult in learning $\mathcal{G}_{\rm LS}$ (See Figure \ref{figure:Lotka-sharpe}). The passage $(k,\mu)\mapsto k\Pi$ is explicit, and the Laplace transform is likewise explicit and linear, even though it acts on an infinite-dimensional input. The central nonlinearity lies in the last step: solving for the unique $\zeta$ such that $\mathcal{L}\{k\Pi\}(\zeta)=1$. In other words, learning $\mathcal{G}_{\rm LS}$ amounts primarily to learning the inverse of the scalar function $\zeta\mapsto \mathcal{L}\{k\Pi\}(\zeta)$ at level $1$. This is a root-finding problem whose solution depends globally on the whole net maternity profile, which explains why $\mathcal{G}_{\rm LS}$ is nontrivial despite the apparent simplicity of its definition.

\subsection{Three easier operators (outputs = direct evaluations of integrals).}
We will see the control law also requires three additional operators - although their mappings are explicit and hence do not require the same computational treatment as the Lotka-Sharpe operator.

\begin{itemize}
    \item $\mathcal{G}_{\kappa}$ : $(k, \mu, \zeta) \mapsto \kappa$, mapping two functions and one scalar into a scalar, and defined explicitly as
\begin{equation}\kappa = \int_0^A a\, k(a)\, e^{-\int_0^a (\zeta+\mu(s))\,ds}\,da\end{equation}
with $\kappa_i = \mathcal{G}_\kappa(k_i, \mu_i, \zeta_i)$, $i=1,2$.

\item $\mathcal{G}_{\gamma}$ : $(g, \zeta, \mu) \mapsto \gamma$, mapping two functions and one scalar into a scalar, and defined explicitly as
\begin{equation}\gamma = \int_0^A g(a)\, e^{-\int_0^a (\zeta+\mu(s))\,ds}\,da\end{equation}
with $\gamma_1 = \mathcal{G}_\gamma(g_1, \zeta_2, \mu_2)$ and $\gamma_2 = \mathcal{G}_\gamma(g_2, \zeta_1, \mu_1)$.

\item $\mathcal{G}_{\pi}$ : $(k, \mu, \zeta) \mapsto \pi_0$, mapping two functions and one scalar into a function, and defined explicitly as
\begin{equation}\pi_0(a) = \int_a^A k(s)\, e^{\int_s^a (\zeta+\mu(l))\,dl}\,ds\end{equation}
with $\pi_{0,i}(a) = \mathcal{G}_\pi(k_i, \mu_i, \zeta_i)(a)$, $i=1,2$.
\end{itemize}

\section{Control Laws: Nominal and Neuro-approximated}

\label{sec-controllers}

We are now ready to introduce the feedback law for stabilizing \eqref{eq:pred-prey-model}. We begin by discussing the exact feedback design and then discuss the design under approximations $\hat{\zeta}_1, \hat{\zeta}_2$ explicitly characterizing the propagation of the error. 

\subsection{Nominal controller ensures stabilization}

\begin{figure*}[t]
\centering
\resizebox{0.72\textwidth}{!}{%
\begin{tikzpicture}[
  >=Latex,
  font=\footnotesize,
  box/.style={
    draw,
    rounded corners=3pt,
    align=center,
    inner sep=3pt,
    minimum height=1.25cm,
    text width=3.2cm
  },
  implicit/.style={
    box,
    draw=blue!70!black,
    fill=blue!12,
    text=blue!60!black
  },
    explicit/.style={
      box,
      draw=teal!70!black,
      fill=teal!12,
      text=teal!60!black,
      text width=2.5cm
    },
  control/.style={
    box,
    draw=gray!70!black,
    fill=gray!15,
    text=black!75,
    text width=4.0cm
  },
  lab/.style={
    font=\scriptsize\itshape,
    text=black!95
  },
  edge/.style={
    draw=black!70,
    line width=0.45pt
  },
  circnum/.style={
  circle,
  draw=black!65,
  line width=0.4pt,
  minimum size=1.55em,
  inner sep=0pt,
  font=\scriptsize
}
]

\node[lab, anchor=east] at (-4.2,  1.9) {\circledlabel{1}{implicit}};
\node[lab, anchor=east] at (-4.2,  0.0) {\circledlabel{2}{explicit}};
\node[lab, anchor=east] at (-4.2, -1.9) {\circledlabel{3}{control}};
\node[implicit] (gls) at (0,1.9) {$\mathcal{G}_{\mathrm{LS}}$\\[1mm] $(k,\mu)\mapsto {\color{red}\zeta}$ implicit};

\node[explicit] (gk)  at (-2.8,0) {$\mathcal{G}_\kappa$\\[1mm] $(k,\mu,{\color{red}\zeta})\mapsto \kappa$};
\node[explicit] (gy)  at (0,0)    {$\mathcal{G}_\gamma$\\[1mm] $(g,{\color{red}\zeta},\mu)\mapsto \gamma$};
\node[explicit] (gpi) at (2.8,0)  {$\mathcal{G}_{\pi}$\\[1mm] $(k,\mu,{\color{red}\zeta})\mapsto \pi_0(\cdot)$};

\node[control] (u) at (0,-1.9) {Control law $u(t)$\\[1mm] Eq. \eqref{eq:nominal-control-law}
};

\draw[edge, ->] (gls.south west) -- ($(gk.north)+(0,0.08)$)
  node[pos=0.55, above, sloped, inner sep=1pt] {${\color{red}\zeta}$};

\draw[edge, ->] (gls.south) -- (gy.north)
  node[pos=0.48, right, inner sep=1pt] {${\color{red}\zeta}$};

\draw[edge, ->] (gls.south east) -- ($(gpi.north)+(0,0.08)$)
  node[pos=0.55, above, sloped, inner sep=1pt] {${\color{red}\zeta}$};

\draw[->, blue!70!black, dashed, line width=0.5pt]
  ($(gls.east)+(0.05,0.15)$)
  .. controls +(0.9,0.35) and +(0.9,-0.35) ..
  ($(gls.east)+(-0.02,-0.15)$);

\node[anchor=west, align=left, text=black!75] at ($(gls.east)+(1.1,0.02)$)
  {${\color{red}\zeta}$ inside\\kernel};

\draw[edge,->] (gk.south) -- ($(u.north west)+(0.2,0)$);
\draw[edge,->] (gy.south) -- (u.north);
\draw[edge,->] (gpi.south) -- ($(u.north east)+(-0.2,0)$);
\node[anchor=west, font=\footnotesize, text=black!95] at (-5.0,-3.0) {%
  \tikz{\fill[blue!20] (0,0) rectangle (0.22,0.14);
        \draw[blue!70!black] (0,0) rectangle (0.22,0.14);}
  \ implicit --- requires iterative root-finding
};

\node[anchor=west, font=\footnotesize, text=black!95] at (2,-3.0) {%
  \tikz{\fill[teal!20] (0,0) rectangle (0.22,0.14);
        \draw[teal!70!black] (0,0) rectangle (0.22,0.14);}
  \ explicit given ${\color{red}\zeta}$
};

\end{tikzpicture}%
}
\caption{Graph of the operator dependencies for the control construction in \eqref{eq-eta-cont}.}
\label{fig:implicit-explicit-operators}
\end{figure*}
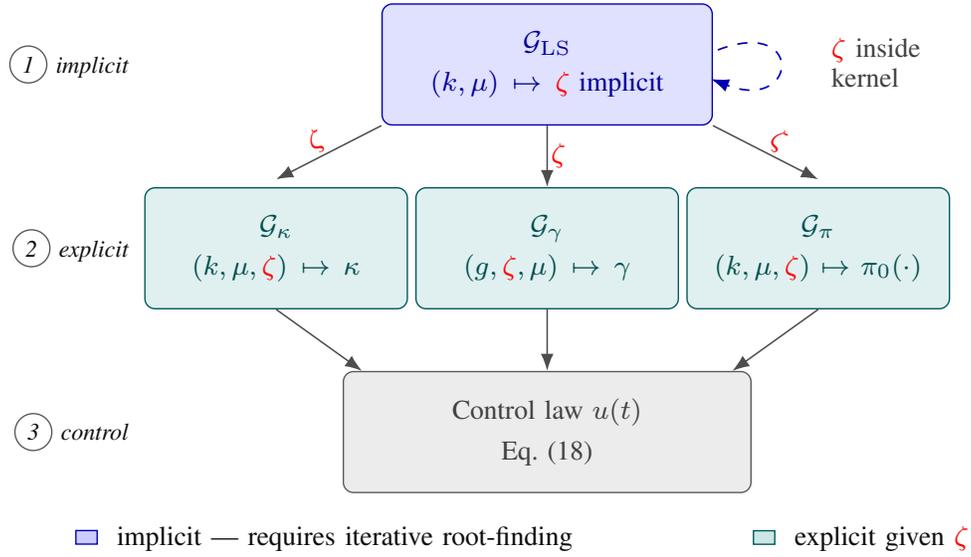

\paragraph*{Nominal controller.} 
The control law
\begin{align}
  \nonumber   u =&\;  u^* + \beta \Bigg[
    \frac{1}{\lambda_2}
    \left(1 - \frac{\int_0^A a k_1(a) x_1^*(a) da}{\int_0^A \pi_{0,1}(a) x_1(a,t)  d a} \right)  
    \\ &- (1+\varepsilon)\lambda_1\left(1 - \frac{\int_0^A \pi_{0,2}(a) x_2(a,t) d a}
    {\int_0^A a k_2(a) x_2^*(a) da} \right)
    \Bigg] 
\end{align}
was designed in \cite{veil2025stabilization} to stabilize system~\eqref{eq:pred-prey-model}.
Using \eqref{eq:def_lambda} to eliminate the constants $\lambda_i$, \eqref{eq:ss_profiles} to eliminate the profiles $x_i^*(a)$, and the relation $u^* = \zeta_2 - \frac{1}{x_1^*(0)\gamma_2}$, the control can be rewritten as
\begin{align}
u =&\; \zeta_2 - \frac{1}{x_1^*(0)\gamma_2}
+ \beta\bigg[\frac{1}{x_1^*(0)\gamma_2}\left(1 - \frac{x_1^*(0)\kappa_1}{\langle\pi_{0,1}, x_1\rangle}\right) \nonumber \\ & \qquad \qquad - (1+\varepsilon)x_2^*(0)\gamma_1\left(1 - \frac{\langle\pi_{0,2}, x_2\rangle}{x_2^*(0)\kappa_2}\right)\bigg]\,,
\end{align}
which, then, eliminating $x_2^*(0)$ using \eqref{eq-x2star-zero} and keeping only the prey birth setpoint $x_1^*(0)$, becomes $u = u_{\rm nom}(\eta)$ with
\begin{align}
u_{\rm nom}(\eta)
:=&\; \zeta_2 - \frac{1}{x_1^*(0)\gamma_2}
+ \beta\bigg[
(1+\varepsilon)(\zeta_2-\zeta_1)
\nonumber \\ &  \qquad -\frac{\varepsilon}{x_1^*(0)\gamma_2}
-\frac{\kappa_1}{\gamma_2\langle\pi_{0,1},x_1\rangle}
\nonumber \\ &  \qquad +(1+\varepsilon)\frac{\gamma_1}{\kappa_2}\langle\pi_{0,2},x_2\rangle
\bigg]. \label{eq:nominal-control-law}
\end{align}
The inputs into this feedback law are the states $(x_1,x_2)$, the setpoint scalar $x_1^*(0)$, as well as the scalars $\zeta_i,\kappa_i,\gamma_i$ and functions $\pi_{0,i}$, which all depend only on $(k_i,\mu_i, g_i)$. 

\paragraph*{Closed-loop stability under nominal/exact $\zeta_1,\zeta_2$.}
It was shown in \cite[Proposition 2]{veil2025stabilization} that, on the set $\{a \in [0, A] | \kappa_i x_i(a) = n_i(a) \langle \pi_{0,i}, x_i\rangle\}$, the PDE system \eqref{eq:pred-prey-model} is governed by the ODE
\begin{subequations}
\label{eq-eta-sys}
\begin{eqnarray}
\dot\eta_1 &=& \zeta_2 - \frac{1}{x_1^*(0)\gamma_2}{\rm e}^{\eta_2} - u\\
\dot\eta_2 &=& \zeta_1 - \frac{1}{x_1^*(0)\gamma_2}{\rm e}^{-\eta_1} - u
\end{eqnarray}
\end{subequations}
We focus in this paper on stabilization of this ODE system, in the presence of approximation errors $\zeta_i-\hat\zeta_i$. 

For exact parameters, $\hat\zeta_i = \zeta_i$, the controller $u_{\rm nom}(\eta)$ gives
\begin{subequations}
\begin{align}
\dot\eta_1 =&\; -\frac{\beta}{x_1^*(0)\gamma_2}(1-e^{-\eta_1}) \nonumber\\ &  - \left(1+\beta(1+\varepsilon)\right)\left(\zeta_1-\zeta_2+\frac{1}{x_1^*(0)\gamma_2}\right)(e^{\eta_2}-1)\\
\dot\eta_2 =&\; \frac{1-\beta}{x_1^*(0)\gamma_2}(1-e^{-\eta_1}) \nonumber \\  &- \beta(1+\varepsilon)\left(\zeta_1-\zeta_2+\frac{1}{x_1^*(0)\gamma_2}\right)(e^{\eta_2}-1)
\end{align}
\end{subequations}

The stability analysis under the nominal controller is conducted with the functions
\begin{subequations}
\begin{align}
\phi_1(\eta_1)=\;&\frac{1}{x_1^*(0)\gamma_2}(1-e^{-\eta_1}), \\
\phi_2(\eta_2)=\;&\left(\zeta_1-\zeta_2+\frac{1}{x_1^*(0)\gamma_2}\right)(e^{\eta_2}-1),
\end{align}
\end{subequations}
where the Lyapunov function is given as
\begin{align}
V_1(\eta)
=\;&
\frac{1}{x_1^*(0)\gamma_2}(e^{-\eta_1}+\eta_1-1)
\nonumber \\ &+
(1+\varepsilon)\left(\zeta_1-\zeta_2+\frac{1}{x_1^*(0)\gamma_2}\right)(e^{\eta_2}-\eta_2-1)
\end{align}
with $\frac{\partial V_1}{\partial \eta_1}=\phi_1(\eta_1), \frac{\partial V_1}{\partial \eta_2}=(1+\varepsilon)\phi_2(\eta_2)$, 
and for the closed-loop system
\begin{subequations}
\begin{eqnarray}
\dot\eta_1 &=& -\beta \phi_1(\eta_1)-\bigl(1+\beta(1+\varepsilon)\bigr)\phi_2(\eta_2), \\
\dot\eta_2 &=& -\beta(1+\varepsilon)\phi_2(\eta_2)+(1-\beta)\phi_1(\eta_1).
\end{eqnarray}
\end{subequations}
The Lyapunov derivative is
\begin{equation}
\dot V_1(\eta)
=
-\begin{bmatrix}\phi_1 & \phi_2\end{bmatrix}
Q
\begin{bmatrix}\phi_1\\ \phi_2\end{bmatrix},
\end{equation}
where
\begin{equation}
Q=
\begin{bmatrix}
\beta & \dfrac{\varepsilon-2\beta(1+\varepsilon)}{2}\\[1.2ex]
\dfrac{\varepsilon-2\beta(1+\varepsilon)}{2} & \beta(1+\varepsilon)^2
\end{bmatrix}\,.
\end{equation}
The determinant 
$\det Q =\frac{\varepsilon\bigl(4\beta(1+\varepsilon)-\varepsilon\bigr)}{4}$ is positive if and only if 
\begin{equation}
\beta>\frac{\varepsilon}{4(1+\varepsilon)}\,,
\end{equation}
which makes $\dot V_1(\eta)$ negative definite and $\eta_1=\eta_2=0$ a globally asymptotically stable equilibrium, at least if $u$ is not restricted to only positive values. 

\subsection{Approximate controller introduces a perturbation}

When the Lotka-Sharpe parameters $\zeta_i$ are known only approximately as $\hat\zeta_i$, $\zeta_i$ is replaced in the controller by $\hat\zeta_i$. The resulting approximation error does not remain localized, but enters every operator that depends on $\zeta_i$. 

The approximation of controller \eqref{eq:nominal-control-law} is given by $u = \hat{u}(\eta, e_1, e_2)$ where 
\begin{align}
\label{eq-eta-cont-PDE}
\hat u(\eta,e_1,e_2)
:=\;& \hat\zeta_2-\frac{1}{x_1^*(0)\hat\gamma_2}
+\beta\bigg[
(1+\varepsilon)(\hat\zeta_2-\hat\zeta_1)
\nonumber \\ & \qquad  -\frac{\varepsilon}{x_1^*(0)\hat\gamma_2}
-\frac{\hat\kappa_1}{\hat\gamma_2\langle \hat\pi_{0,1},x_1\rangle}
\nonumber \\ &\qquad  +(1+\varepsilon)\frac{\hat\gamma_1}{\hat\kappa_2}\langle \hat\pi_{0,2},x_2\rangle
\bigg]
\\
=\;& \hat\zeta_2-\frac{1}{x_1^*(0)\hat\gamma_2}
+\beta\bigg[
(1+\varepsilon)(\hat\zeta_2-\hat\zeta_1)
\nonumber \\ & \qquad -\frac{\varepsilon}{x_1^*(0)\hat\gamma_2}
-\frac{\hat\kappa_1}
{\hat\gamma_2\,\langle \hat\pi_{0,1},n_1\rangle\,x_1^*(0)}e^{-\eta_1}
\nonumber \\ & \qquad +(1+\varepsilon)\frac{\hat\gamma_1\,\langle \hat\pi_{0,2},n_2\rangle\,x_2^*(0)}
{\hat\kappa_2}e^{\eta_2}
\bigg]
\,,
\label{eq-eta-cont}
\end{align}
where
\begin{equation}
\eta_i = \ln\left(\frac{\langle \pi_{0,i}, x_i\rangle}{x_i^*(0) \kappa_i}\right)\,, 
\end{equation}
and
\begin{eqnarray}
\zeta_i &=&  \mathcal{G}_{\rm LS}( k_i,\mu_i) \\
\hat\zeta_i &=& \widehat{\mathcal{G}}_{\rm LS}( k_i,\mu_i) = \zeta_i-e_i\\
e_i &=& \zeta_i-\hat\zeta_i\\
\label{eq-kappahat}
\hat\kappa_i &=& 
\mathcal{G}_\kappa( k_i,  \mu_i,\zeta_i-e_i) \\
\label{eq-gamma1hat}
\hat\gamma_1 &=& \mathcal{G}_\gamma( g_1, \zeta_2-e_2, \mu_2)\\
\label{eq-gamma2hat}
\hat\gamma_2 &=& \mathcal{G}_\gamma( g_2, \zeta_1-e_1, \mu_1)\\
\label{eq-pihat}
\hat \pi_{0,i} &=& \mathcal{G}_\pi(k_i, \mu_i, \zeta_i-e_i)\,,
\end{eqnarray}
where $\widehat{\mathcal{G}}_{\rm LS}$ stands for an approximation of the operator ${\mathcal{G}}_{\rm LS}$, which can be of a neural, numerical, or another kind, producing errors $e_i$. 

In \eqref{eq-kappahat}-\eqref{eq-pihat} one faces a nearly terrifying feature of the approximate controller: the approximation errors $e_i$ of the Lotka-Sharpe operator propagate throughout the gain architecture of the control law. The robustness analysis will have to quantify all of them, through the respective nonlinear infinite-dimensional operators $\mathcal{G}_\kappa, \mathcal{G}_\gamma, \mathcal{G}_\pi$. 

Further, note explicitly that $u_{\rm nom}(\eta) = \hat u(\eta,0,0)$.
Dealing with the perturbation 
\begin{equation}
\Delta_u(\eta,e_1,e_2) = \hat u(\eta,e_1,e_2) - \hat u(\eta,0,0)
\end{equation}
 is the main technical challenge to be overcome in the stability analysis portion of this paper. Additionally, 
though negative definite, $\dot V_1(\eta)$ is not proper. The lack of properness is the model's fundamental challenge for achieving semiglobal stability in the face of the controller perturbation $\Delta_u(\eta,e_1,e_2)$. Theorem \ref{thm-spas} is where these challenges are overcome.

\def\nodescale{0.82}   
\begin{figure*}[t]
    \centering 
\resizebox{500pt}{!}{
\begin{tikzpicture}[
    x=1cm,y=1cm,
    >=Latex,
    font=\sffamily,
    every node/.style={transform shape, scale=\nodescale},
    line cap=round,
    line join=round,
    inarrow/.style={->, draw=black!60, line width=1.0pt},
    meas/.style={->, draw=black!55, dashed, line width=1.0pt},
    zeta/.style={->, draw=blue!55, line width=1.5pt},
    gyline/.style={->, draw=red!65!black, line width=1.0pt},
    gkline/.style={->, draw=green!45!black, line width=1.0pt},
    gpiline/.style={->, draw=orange!75!black, line width=1.0pt},
    block/.style={
      draw,
      rounded corners=7pt,
      minimum width=2.45cm,
      minimum height=1.55cm,
      align=center,
      line width=1.0pt
    },
    ls/.style={
      block,
      draw=blue!70!black,
      fill=blue!10,
      text=blue!75!black,
      line width=1.6pt,
      rounded corners=9pt
    },
    gk/.style={block, draw=green!45!black, fill=green!12, text=green!35!black},
    gy/.style={block, draw=red!55!black, fill=red!8, text=red!45!black},
    gpi/.style={block, draw=orange!65!black, fill=orange!14, text=orange!45!black},
    control/.style={
      draw=black!65,
      fill=gray!10,
      rounded corners=8pt,
      minimum width=7.6cm,
      minimum height=1.35cm,
      align=center,
      line width=1.1pt,
      text=black
    },
    zetaLine/.style={draw=blue!55, line width=1.5pt},   
    muLine/.style={draw=black!60, line width=1.0pt},
    kLine/.style={draw=black!60, line width=1.0pt}
]


\def\muxL{-1.6}
\def\muxR{1.6}
\def\kshift{3.0}
\def\gshift{7.0}

\def\xgL{\muxL-\gshift}
\def\xkL{\muxL-\kshift}
\def\xmuL{\muxL}

\def\xmuR{\muxR}
\def\xkR{\muxR+\kshift}
\def\xgR{\muxR+\gshift}

\def\ytoplabel{9.1}
\def\ytopstart{8.8}
\def\yshort{6.95}
\def\ylong{4.45}

\node[font=\Large] at (\xmuL,\ytoplabel) {$k_1(a)$};
\node[font=\Large] at (\xmuR,\ytoplabel) {$k_2(a)$};

\node[font=\Large] at (\xkL,\ytoplabel) {$\mu_1(a)$};
\node[font=\Large] at (\xgL,\ytoplabel) {$g_2(a)$};

\node[font=\Large] at (\xkR,\ytoplabel) {$\mu_2(a)$};
\node[font=\Large] at (\xgR,\ytoplabel) {$g_1(a)$};

\node[ls, minimum width=4cm, minimum height=1.65cm] (LS1) at (-3.2,7)
{$\widehat{\mathcal{G}}_{\mathrm{LS}}$\\[2pt]
$\left(k_1,\mu_1\right)\mapsto {\color{red}\hat\zeta_1}$};

\node[ls, minimum width=4cm, minimum height=1.65cm] (LS2) at (3.2, 7)
{$\widehat{\mathcal{G}}_{\mathrm{LS}}$\\[2pt]
$\left(k_2,\mu_2\right)\mapsto {\color{red}\hat\zeta_2}$};

\def\yLSbend{8.5}
\def\LSinsep{0.55}

\coordinate (LS1K) at ($(LS1.north)+(-\LSinsep,0)$);
\coordinate (LS2K) at ($(LS2.north)+(\LSinsep,0)$);

\path let \p1 = (LS1.north), \p2 = (LS2.north) in
  coordinate (LS1Mu) at (\xmuL,\y1)
  coordinate (LS2Mu) at (\xmuR,\y2);


\draw[inarrow] (\xkL,\ytopstart) -- (\xkL,\yLSbend) -| (LS1K);
\draw[inarrow] (\xkR,\ytopstart) -- (\xkR,\yLSbend) -| (LS2K);

\def\dx{3.2}   

\node[gy] (GY2) at (-2.5*\dx,3.6)
{$\mathcal{G}_{\gamma}$\\[2pt]
$\left(g_2,{\color{red}\hat\zeta_1},\mu_1\right)\mapsto \gamma_2$};

\node[gk] (GK1) at (-1.5*\dx,3.6)
{$\mathcal{G}_{\kappa}$\\[2pt]
$\left(k_1,\mu_1,{\color{red}\hat\zeta_1}\right)\mapsto \kappa_1$};

\node[gpi] (GPI1) at (-0.5*\dx,3.6)
{$\mathcal{G}_{\pi}$\\[2pt]
$\left(k_1,\mu_1,{\color{red}\hat\zeta_1}\right)\mapsto \pi_{0,1}$};

\node[gpi] (GPI2) at (0.5*\dx,3.6)
{$\mathcal{G}_{\pi}$\\[2pt]
$\left(k_2,\mu_2,{\color{red}\hat\zeta_2}\right)\mapsto \pi_{0,2}$};

\node[gk] (GK2) at (1.5*\dx,3.6)
{$\mathcal{G}_{\kappa}$\\[2pt]
$\left(k_2,\mu_2,{\color{red}\hat\zeta_2}\right)\mapsto \kappa_2$};

\node[gy] (GY1) at (2.5*\dx,3.6)
{$\mathcal{G}_{\gamma}$\\[2pt]
$\left(g_1,{\color{red}\hat\zeta_2},\mu_2\right)\mapsto \gamma_1$};

\node[control] (CTL) at (0,-0.05)
{$\mathbf{\text{Control law }}$
Eq. \eqref{eq-eta-cont}};

\def\mybus{5.7}   


\path let \p1 = (LS1.north), \p2 = (LS2.north) in
  coordinate (LS1Mu) at (\xkL,\y1)
  coordinate (LS2Mu) at (\xkR,\y2);

\coordinate (MuLTee) at (\xkL,\yLSbend);
\coordinate (MuRTee) at (\xkR,\yLSbend);

\def\GYtwoMuX{0.6}
\def\GKoneMuX{-0.6}
\def\GPIoneMuX{-0.6}

\def\GPItwoMuX{0.6}
\def\GKtwoMuX{0.6}
\def\GYoneMuX{-0.6}

\coordinate (GY2mu)  at ($(GY2.north)+(\GYtwoMuX,0)$);
\coordinate (GK1mu)  at ($(GK1.north)+(\GKoneMuX,0)$);
\coordinate (GPI1mu) at ($(GPI1.north)+(\GPIoneMuX,0)$);

\coordinate (GPI2mu) at ($(GPI2.north)+(\GPItwoMuX,0)$);
\coordinate (GK2mu)  at ($(GK2.north)+(\GKtwoMuX,0)$);
\coordinate (GY1mu)  at ($(GY1.north)+(\GYoneMuX,0)$);

\path let \p1 = (GY2.north) in coordinate (GY2topx) at (\xgL,\y1);
\path let \p1 = (GY1.north) in coordinate (GY1topx) at (\xgR,\y1);

\draw[inarrow] (\xgL,\ytopstart) -- (GY2topx);
\draw[inarrow] (\xgR,\ytopstart) -- (GY1topx);

\path let
  \p1 = (GY2mu),
  \p2 = (GK1mu),
  \p3 = (GPI1mu),
  \p4 = (GPI2mu),
  \p5 = (GK2mu),
  \p6 = (GY1mu)
in
  coordinate (ML1) at (\x1,\mybus)
  coordinate (ML2) at (\x2,\mybus)
  coordinate (ML3) at (\x3,\mybus)
  coordinate (MR1) at (\x4,\mybus)
  coordinate (MR2) at (\x5,\mybus)
  coordinate (MR3) at (\x6,\mybus);

\def\muPushL{1.5}
\def\muPushR{1.5}

\coordinate (MuLElbow) at ($(MuLTee)+(-\muPushL,0)$);
\coordinate (MuRElbow) at ($(MuRTee)+(\muPushR,0)$);

\draw[muLine] (MuLTee) -- (MuLElbow) -- (MuLElbow |- ML2) -- (ML2);
\draw[muLine] (MuRTee) -- (MuRElbow) -- (MuRElbow |- MR2) -- (MR2);

\draw[muLine] (ML1) -- (ML3);
\draw[muLine] (MR1) -- (MR3);

\draw[inarrow] (ML1) -- (GY2mu);
\draw[inarrow] (ML2) -- (GK1mu);
\draw[inarrow] (ML3) -- (GPI1mu);

\draw[inarrow] (MR1) -- (GPI2mu);
\draw[inarrow] (MR2) -- (GK2mu);
\draw[inarrow] (MR3) -- (GY1mu);
\def\kybus{5.35}      
\def\kPushL{1.15}     
\def\kPushR{1.15}     

\coordinate (KLTee) at (\xmuL,\yLSbend);
\coordinate (KRTee) at (\xmuR,\yLSbend);

\coordinate (KLElbow) at ($(KLTee)+(\kPushL,0)$);
\coordinate (KRElbow) at ($(KRTee)+(-\kPushR,0)$);
\def\kPushLOut{1}     
\def\kPushROut{1}
\coordinate (KLElbowOut) at ($(KLTee)-(\kPushLOut,0)$);
\coordinate (KRElbowOut) at ($(KRTee)-(-\kPushROut,0)$);

\def\GKoneKX{0.6}
\def\GPIoneKX{0.6}

\def\GPItwoKX{-0.6}
\def\GKtwoKX{-0.6}

\coordinate (GK1k)  at ($(GK1.north)+(\GKoneKX,0)$);
\coordinate (GPI1k) at ($(GPI1.north)+(\GPIoneKX,0)$);

\coordinate (GPI2k) at ($(GPI2.north)+(\GPItwoKX,0)$);
\coordinate (GK2k)  at ($(GK2.north)+(\GKtwoKX,0)$);

\path let
  \p1 = (GK1k),
  \p2 = (GPI1k),
  \p3 = (GPI2k),
  \p4 = (GK2k)
in
  coordinate (KL1) at (\x1,\kybus)
  coordinate (KL2) at (\x2,\kybus)
  coordinate (KR1) at (\x3,\kybus)
  coordinate (KR2) at (\x4,\kybus);

\draw[kLine] (KLTee) -- (KLElbow);
\draw[kLine] (KRTee) -- (KRElbow);

\draw[kLine] (\xmuL,\ytopstart) -- (KLTee);
\draw[kLine] (\xmuR,\ytopstart) -- (KRTee);

\draw[kLine] (KLTee) -- (KLElbowOut);
\draw[kLine] (KRTee) -- (KRElbowOut);

\draw[inarrow] (KLElbowOut) -- ($(LS1.north)+(0.6,0)$); 
\draw[inarrow] (KRElbowOut) -- ($(LS2.north)+(-0.6,0)$); 

\draw[kLine] (KLElbow) -- (KLElbow |- KL1) -- (KL1);
\draw[kLine] (KRElbow) -- (KRElbow |- KR2) -- (KR2);

\draw[kLine] (KL1) -- (KL2);
\draw[kLine] (KR1) -- (KR2);

\draw[inarrow] (KL1) -- (GK1k);
\draw[inarrow] (KL2) -- (GPI1k);

\draw[inarrow] (KR1) -- (GPI2k);
\draw[inarrow] (KR2) -- (GK2k);

\def\zybus{5.}   

\def\GYtwoX{0}
\def\GKoneX{0.00}
\def\GPIoneX{0}

\def\GPItwoX{0}
\def\GKtwoX{0.00}
\def\GYoneX{0}

\coordinate (GY2in)  at ($(GY2.north)+(\GYtwoX,0)$);
\coordinate (GK1in)  at ($(GK1.north)+(\GKoneX,0)$);
\coordinate (GPI1in) at ($(GPI1.north)+(\GPIoneX,0)$);

\coordinate (GPI2in) at ($(GPI2.north)+(\GPItwoX,0)$);
\coordinate (GK2in)  at ($(GK2.north)+(\GKtwoX,0)$);
\coordinate (GY1in)  at ($(GY1.north)+(\GYoneX,0)$);

\path let
  \p1 = (GY2in),
  \p2 = (GK1in),
  \p3 = (GPI1in),
  \p4 = (GPI2in),
  \p5 = (GK2in),
  \p6 = (GY1in)
in
  coordinate (ZL1) at (\x1,\zybus)
  coordinate (ZL2) at (\x2,\zybus)
  coordinate (ZL3) at (\x3,\zybus)
  coordinate (ZR1) at (\x4,\zybus)
  coordinate (ZR2) at (\x5,\zybus)
  coordinate (ZR3) at (\x6,\zybus);

\coordinate (LS1bus) at (-2.55,\zybus);
\coordinate (LS2bus) at ( 2.55,\zybus);

\draw[zetaLine] (LS1.south) |- (LS1bus);
\draw[zetaLine] (LS2.south) |- (LS2bus);

\draw[zetaLine] (ZL1) -- (ZL3);
\draw[zetaLine] (ZR1) -- (ZR3);

\draw[zeta] (ZL1) -- (GY2in);
\draw[zeta] (ZL2) -- (GK1in);
\draw[zeta] (ZL3) -- (GPI1in);

\draw[zeta] (ZR1) -- (GPI2in);
\draw[zeta] (ZR2) -- (GK2in);
\draw[zeta] (ZR3) -- (GY1in);

\node[font=\Large, text=blue!60] 
    at ($(LS1bus)+(-0.7,-0.36)$) {$\boldsymbol{{\color{red}\hat\zeta_1}}$};

\node[font=\Large, text=blue!60] 
    at ($(LS2bus)+(0.7,-0.36)$) {$\boldsymbol{{\color{red}\hat\zeta_2}}$};

\node[font=\Large, text=red!65!black]    at (-7.5,1.9) {$\gamma_2$};
\node[font=\Large, text=green!40!black]  at (-3.95,1.83) {$\kappa_1$};
\node[font=\Large, text=orange!70!black] at (-1.35,1.9) {$\pi_{0,1}$};

\node[font=\Large, text=orange!70!black] at (1.35,1.9) {$\pi_{0,2}$};
\node[font=\Large, text=green!40!black]  at (3.95,1.83) {$\kappa_2$};
\node[font=\Large, text=red!65!black]    at (7.5,1.9) {$\gamma_1$};

\def\yctrl{0.9}     
\def\dropA{1.4}
\def\dropB{0.7}
\def\dropC{0.2}

\coordinate (C1) at (-2.8,\yctrl);
\coordinate (C2) at (-1.9,\yctrl);
\coordinate (C3) at (-0.9,\yctrl);
\coordinate (C4) at ( 0.9,\yctrl);
\coordinate (C5) at ( 1.9,\yctrl);
\coordinate (C6) at ( 2.8,\yctrl);

\draw[gyline]  (GY2.south)  -- ++(0,-\dropA) -| (C1) -- (C1 |- CTL.north);
\draw[gkline]  (GK1.south)  -- ++(0,-\dropB) -| (C2) -- (C2 |- CTL.north);
\draw[gpiline] (GPI1.south) -- ++(0,-\dropC) -| (C3) -- (C3 |- CTL.north);

\draw[gpiline] (GPI2.south) -- ++(0,-\dropC) -| (C4) -- (C4 |- CTL.north);
\draw[gkline]  (GK2.south)  -- ++(0,-\dropB) -| (C5) -- (C5 |- CTL.north);
\draw[gyline]  (GY1.south)  -- ++(0,-\dropA) -| (C6) -- (C6 |- CTL.north);
\draw[meas] (-7,-0.05) -- (CTL.west);
\node[anchor=east, font=\Large] at (-7.05,-0.05) {$x_1(a,t)$};

\draw[meas] (7,-0.05) -- (CTL.east);
\node[anchor=west, font=\Large] at (7.05,-0.05) {$x_2(a,t)$};

\draw[->, draw=black!70, line width=1.2pt] (CTL.south) -- ++(0,-0.55);
\node[anchor=west, font=\Large] at (0.12,-1.05) {$u(t)$};

\draw[inarrow] (-6.7,-1.85) -- (-5.8,-1.85);
\node[anchor=west, font=\Large, text=black!80] at (-5.7,-1.85) {input function};

\draw[zeta] (-2.25,-1.85) -- (-1.35,-1.85);
\node[anchor=west, font=\Large, text=black!80] at (-1.25,-1.85) {$\zeta$ output};

\draw[meas] (1.2,-1.85) -- (2.1,-1.85);
\node[anchor=west, font=\Large, text=black!80] at (2.2,-1.85) {given / measured};

\draw[dashed, ->, draw=black!45, line width=0.9pt] (0,9.3) -- (CTL.north);
\node[anchor=south, font=\Large] at (0.05,9.4) {$x_1^{*}(0)$};

\end{tikzpicture}
}
\caption{Explicit operator mappings in the predator-prey control law.}
\end{figure*}
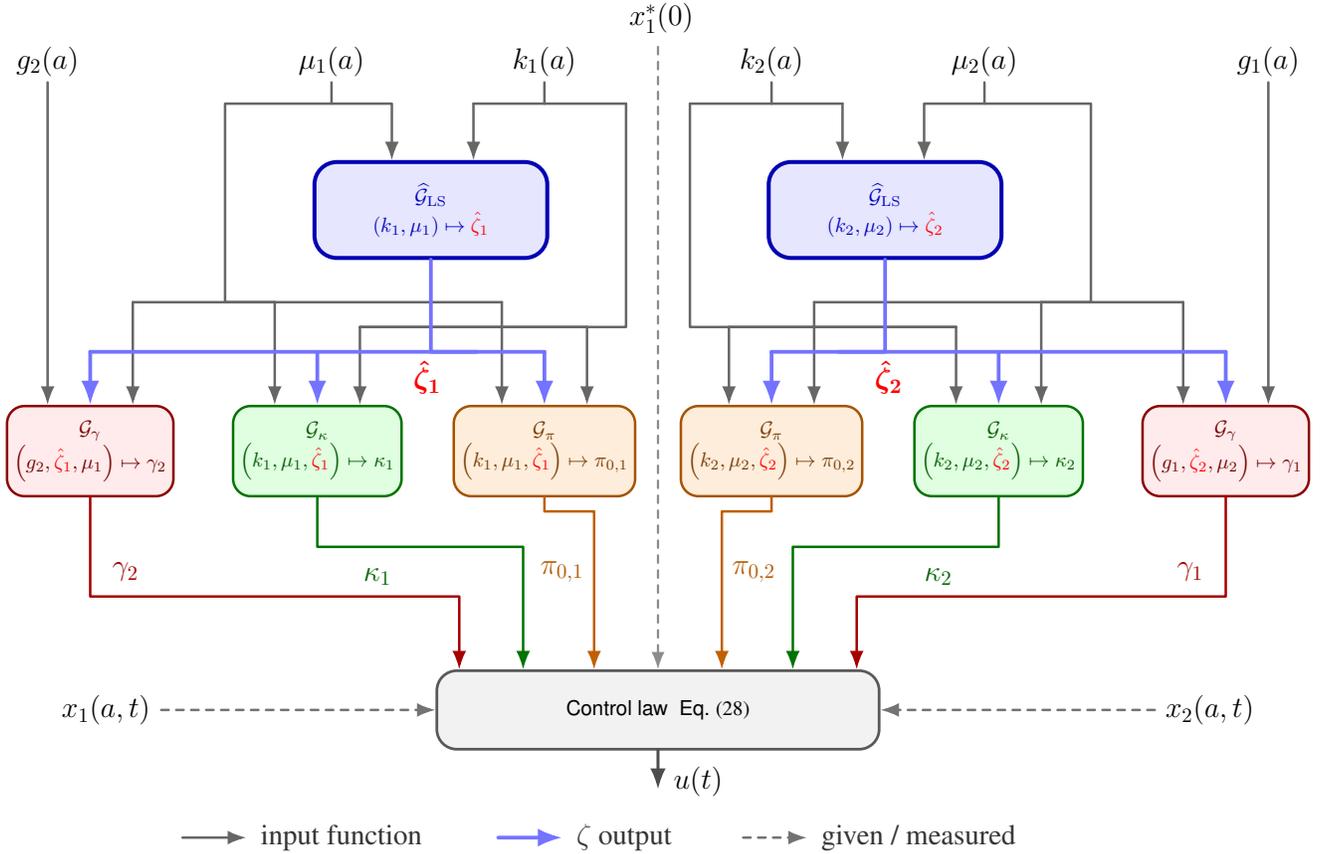

\section{Neural Approximability of Lotka-Sharpe Operator}
\label{sec-approx}

Relative to many previous results on neural operator-based control, what differentiates the result of the present paper is the operator: the Lotka-Sharpe nonlinear mapping. The next theorem is the backbone of the paper --- establishing Lipschitzness of $\mathcal{G}_{\rm LS}$. The result is unconventional: from the domain on which the result holds, whose idiosyncrasy comes from biology-required monotonicity of birth and mortality, to the technique with which the Lipschitz constant is derived. 

\begin{theorem} \emph{(Lipschitz continuity of the Lotka--Sharpe mapping)} 
\label{thm-LS-Lip}
Let $A>0$ and let $k_{\rm min},k_{\rm max},\mu_{\rm min},\mu_{\rm max} \in C^0([0,A];\mathbb{R}_{\ge 0})$ be Lipschitz functions such that
\begin{align}
\label{eq-k1k2}
k_{\rm min}(a)&\le k_{\rm max}(a),  \nonumber \\  \mu_{\rm min}(a)&\le \mu_{\rm max}(a),
\quad \forall a\in[0,A],
\end{align}
and
\begin{equation}
\label{eq-int>1}
\int_0^A k_{\rm min}(a)e^{-\int_0^a \mu_{\rm max}(s)\,ds}\,da > 1.
\end{equation}
Let $\zeta_{\rm min}\leq \zeta_{\rm max}$ be the unique solutions of the following equations for ($k_{\rm min}$,$\mu_{\rm max}$) and ($k_{\rm max}$, $\mu_{\rm min}$) respectively. Then, 
\begin{align}
\label{eq-LS-map-1-2}
\int_0^A k_{\rm min}(a)e^{-\zeta_{\rm min} a-\int_0^a \mu_{\rm max}(s)\,ds}\,da &= 1\,, \\ 
\int_0^A k_{\rm max}(a)e^{-\zeta_{\rm max} a-\int_0^a \mu_{\rm min}(s)\,ds}\,da &= 1.
\end{align}
Let $G>0$ be a sufficiently large constant so that
\begin{equation}
\|f\|_\infty + \sup_{\substack{a,s\in[0,A]\\ a\ne s}} \frac{|f(a)-f(s)|}{|a-s|} \le G
\end{equation}
for $f=k_{\rm min},k_{\rm max},\mu_{\rm min},\mu_{\rm max}$. Define
\begin{equation}
\label{eq:set_HG}
H_G := \left\{
f \in C^0[0,A];\mathbb{R}_{\ge 0}) :
\|f\|_\infty + [f]_{C^{0,1}([0,A])} \le G
\right\},
\end{equation}
where 
\begin{equation}
[f]_{C^{0,1}([0,A])}
:=
\sup_{\substack{a,s\in[0,A]\\ a\ne s}}
\frac{|f(a)-f(s)|}{|a-s|}\,, 
\end{equation}
and
\begin{equation}
\label{eq:set_S}
S := \left\{
\begin{aligned}
(k,\mu)\in H_G^2 \;:\;& k_{\rm min}(a)\le k(a)\le k_{\rm max}(a),\\
& \mu_{\rm min}(a)\le \mu(a)\le \mu_{\rm max}(a),\\
& \text{for all } a\in[0,A]
\end{aligned}
\right\}.
\end{equation}
Furthermore, for each $(k,\mu)\in S$, define $P(k,\mu)=\zeta$, where $\zeta>0$ is the unique solution of
\begin{equation}
\label{eq-LS-map-def}
\int_0^A k(a)e^{-\zeta a-\int_0^a \mu(s)\,ds}\,da = 1.
\end{equation}
Then the mapping $P:S\to[\zeta_{\rm min},\zeta_{\rm max}]$ is Lipschitz continuous with respect to the sup norm, namely,
for all $(k,\mu),(\tilde k,\tilde\mu)\in S$,
\begin{equation}
|P(\tilde k,\tilde\mu)-P(k,\mu)|
\le
L\,\|\tilde k-k\|_\infty
+
L\,\|k_{\rm max}\|_\infty\,A\,\|\tilde\mu-\mu\|_\infty,
\end{equation}
where
\begin{align}
\label{eq-Lip-const}
L =\;&
\frac{
A\,\bigl(2A\|k_{\rm max}\|_\infty\bigr)^{2A\|k_{\rm max}\|_\infty-1}
}{
\left(\int_0^A a\,k_{\rm min}(a)I(a)\,da\right)
\ln\!\left(\int_0^A k_{\rm min}(a)I(a)\,da\right)
}\,, 
\end{align}
with
\begin{align}
    I(a) :=\; e^{-\int_0^a \mu_{\rm max}(s) ds }\,. 
\end{align}
\end{theorem}

As depicted in Figure \ref{figure:Lotka-sharpe}, the Lotka–Sharpe operator reduces to a scalar root-finding problem for the Laplace transform of the net maternity function $k(a)\Pi(a)$. The only nonlinear step is the function inversion $F(\zeta)=1$; all other steps are explicit. The analysis below quantifies the {\em sensitivity of this function inversion} to perturbations in the input functions $(k,\mu)$.

The Lipschitz constant \eqref{eq-Lip-const} reflects the sensitivity of the intrinsic growth rate to perturbations in fertility and mortality profiles, increasing when reproduction is either high (large $k_{\rm max}$) or when mortality-discounted fertility is weak (small $k_{\rm min}$, large $\mu_{\rm max}$), i.e., when the population operates near a fragile balance between growth and decline. 

From the constant \eqref{eq-Lip-const} being independent of the Lipschitz constants of $k_{\rm min}, k_{\rm max}, \mu_{\rm min}, \mu_{\rm max}$ the reader should not infer that $S$ can be a set that contains non-Lipschitz continuous functions. In order to obtain the
neural approximability $S$ needs to be a compact set, which is ensured by the definitions of $H_G$ and $S$ and by the Arzela-Ascoli theorem.\footnote{A slight
generalization can be obtained if the set $H_G$ is not a bounded subset of Lipschitz functions but a bounded subset of Holder continuous functions (where equicontinuity holds as well), but we forego that generalization.}

The next result, an immediate corollary of the universal approximation theory for nonlinear operators on compact domains (See \cite{chen1995Universal},\cite[Theorem 1]{Lanthaler2025Nonlocality}), provides the neural approximation mechanism that is used to replace the exact Lotka–Sharpe parameters in the feedback law. 

\begin{corollary} \emph{(Neural-operator approximability of the Lotka--Sharpe mapping)} 
Let $A>0$, and let $k_{\rm min},k_{\rm max},\mu_{\rm min},\mu_{\rm max}\in C^0([0,A];\mathbb R_{\ge 0})$ satisfy \eqref{eq-k1k2} and \eqref{eq-int>1}. 
Let $G>0$, define $H_G$ and $S\subset H_G^2$ as in Theorem \ref{thm-LS-Lip}.
Define the Lotka--Sharpe operator as $P$ in Theorem \ref{thm-LS-Lip}, namely, as
\begin{equation}
\mathcal{G}_{\rm LS}:S\to[\zeta_{\rm min},\zeta_{\rm max}],\qquad \mathcal{G}_{\rm LS}(k,\mu):=\zeta,
\end{equation}
where $\zeta$ is the unique solution of \eqref{eq-LS-map-def} and $\zeta_{\rm \min} \leq \zeta_{\rm max}$ be the constants defined by \eqref{eq-LS-map-1-2}.
Then, for every $\delta>0$, there exists a neural operator $\widehat{\mathcal{G}}_{\rm LS}:S\to\mathbb R$ such that
\begin{equation}
|\mathcal{G}_{\rm LS}(k,\mu)-\widehat{\mathcal{G}}_{\rm LS}(k,\mu)|\le \delta,\qquad \forall\,(k,\mu)\in S.
\end{equation}
\end{corollary}

\begin{proof}
By Theorem \ref{thm-LS-Lip}, $\mathcal{G}_{\rm LS}$ is Lipschitz on $S$, hence continuous on the compact set $S$. The universal approximation property of neural operators \cite{lu2021Learning} implies that $\mathcal{G}_{\rm LS}$ can be uniformly approximated on $S$ by a neural operator $\widehat{\mathcal{G}}_{\rm LS}$ with arbitrary accuracy $\delta>0$.
\end{proof}

\section{Stabilization with Neural Operator}
\label{sec-stabilization}

\subsection{Stability theorem under errors in Lotka-Sharpe parameters $\zeta_1,\zeta_2$}

We first state a generic approximation-robustness result, in Theorem \ref{thm-spas}, which is independent of whether the approximation originates from a neural operator or some other error in computing the Lotka-Sharpe parameters $\zeta_1,\zeta_2$. Then, in Corollary \ref{cor-main}, we give a stabilization result under a neural operator. 

Note that the approximate controller \eqref{eq-eta-cont} depends on the approximation $\hat\zeta_i=\zeta_i - e_i$ not only directly, but also through the derived quantities $\hat\kappa_i, \hat\gamma_i, \hat\pi_{0,i}$ defined in \eqref{eq-kappahat}, \eqref{eq-gamma1hat}, \eqref{eq-gamma2hat}, \eqref{eq-pihat}, which are exact evaluations of $\mathcal{G}_\kappa, \mathcal{G}_\gamma, \mathcal{G}_\pi$ at the approximate Lotka-Sharpe values. Consequently, all approximation errors in the controller reduce to the scalar errors $e_1, e_2$  and Theorem \ref{thm-spas} is stated entirely in terms of these.

Even though the (approximate) feedback law~\eqref{eq-eta-cont} is given in terms of the state of the PDE, see the version \eqref{eq-eta-cont-PDE}, we provide a stability guarantee for the reduced/ODE model~\eqref{eq-eta-sys}. We have two reasons for this. First, the full PDE model is equivalent to the ODE model with exponentially decaying multiplicative perturbations, as shown in \cite[(15)]{veil2025stabilization}, and noting the exponential decay of $\psi_i$ in \cite[(53), (54)]{veil2025stabilization}. We can extend our Theorem \ref{thm-spas} here just as we extended the ODE Theorem 1 to the PDE Theorem 2 in \ref{thm-spas} in \cite{veil2025stabilization}. Second, such an extension would not illuminate --- it would, in fact, detract from the clarity of how the errors of approximating the Lotka-Sharpe and the other three operators are handled in our robustness analysis.

\begin{theorem}
(\emph{Admissibly semi-global practical asymptotic stability under positive control})
\label{thm-spas}
Consider system~\eqref{eq-eta-sys} controlled by the approximate feedback law~\eqref{eq-eta-cont}. Let $\varepsilon>0$,
\begin{equation}
\beta>\frac{\varepsilon}{4(1+\varepsilon)},
\end{equation}
$\gamma_1,\gamma_2,\zeta_1,\zeta_2>0$, $x_1^*(0) > \frac{1}{\zeta_2\gamma_2 }$, $x_2^*(0)>0$, and define
\begin{equation}
a:=\frac{1}{x_1^*(0)\gamma_2},\qquad
b:=\zeta_1-\zeta_2+\frac{1}{x_1^*(0)\gamma_2} = \gamma_1 x_2^*(0) > 0\
\end{equation}
\begin{equation}
\phi_1(\eta_1):=a(1-e^{-\eta_1}),\qquad
\phi_2(\eta_2):=b(e^{\eta_2}-1),
\end{equation}
\begin{equation}
r(\eta):=\sqrt{\phi_1(\eta_1)^2+\phi_2(\eta_2)^2},
\end{equation}
\begin{equation}
\mathcal D_*:=\{\eta\in\mathbb R^2:\ r(\eta)<\min\{a,b\}\},
\end{equation}
and
\begin{equation}
V_1(\eta):=a\bigl(e^{-\eta_1}+\eta_1-1\bigr)
+(1+\varepsilon)b\bigl(e^{\eta_2}-\eta_2-1\bigr).
\end{equation}
Define
\begin{equation}
\Omega_c:=\{\eta\in\mathbb R^2:\ V_1(\eta)\le c\},
\end{equation}
and for $\delta>0$
\begin{align}
\label{eq:cdelta-star}
c_\delta^*
&:= \sup\Bigl\{ c>0 : \Omega_c\subset \mathcal D_*
\ \text{and}
\nonumber\\
&\hphantom{{}:= \sup\Bigl\{ c>0 : {}}
\inf_{\eta\in\Omega_c,\ |e_1|+|e_2|\le\delta}
u(\eta;e_1,e_2)\ge 0 \Bigr\}.
\end{align}
Then, for every $\delta>0$ and every $c\in(0,c_\delta^*)$, there exist functions $\beta_c\in\mathcal{K}\mathcal{L}$ and $\mu_c\in\mathcal K$ such that if
\begin{equation}
|e_1|+|e_2|\le\delta,\qquad \eta(0)\in\Omega_c,
\end{equation}
 the solution satisfies
\begin{equation}
\eta(t)\in\Omega_c\subset\mathcal D_*,\qquad u(t)\ge 0,\qquad \forall t\ge 0,
\end{equation}
and
\begin{equation}
r(\eta(t))\le \beta_c\bigl(r(\eta(0)),t\bigr)+\mu_c(\delta),\qquad \forall t\ge 0.
\end{equation}
\end{theorem}

\subsection{Main result---stabilization under Lotka-Sharpe NO}

\begin{corollary}\emph{(Admissibly semi-global practical asymptotic stability under neural approximation)} 
\label{cor-main}
Let the assumptions of Theorem~\ref{thm-spas} hold. Consider system~\eqref{eq-eta-sys} controlled by the approximate feedback law~\eqref{eq-eta-cont}, and assume that the neural operator $\widehat{\mathcal{G}}_{\rm LS}$ satisfies
\begin{eqnarray}
\bigl|\mathcal{G}_{\rm LS}(k_1,\mu_1)-\widehat{\mathcal{G}}_{\rm LS}(k_1,\mu_1)\bigr|&<&\frac{\delta}{2}, \\ 
\bigl|\mathcal{G}_{\rm LS}(k_2,\mu_2)-\widehat{\mathcal{G}}_{\rm LS}(k_2,\mu_2)\bigr|&<&\frac{\delta}{2},
\end{eqnarray}
or, equivalently, 
\begin{equation}
e_1:=\zeta_1-\hat\zeta_1,\qquad e_2:=\zeta_2-\hat\zeta_2,
\end{equation}
satisfy $|e_1|<\delta/2$, $|e_2|<\delta/2$, and, combined,  $|e_1|+|e_2|<\delta$. Then, for every $c\in(0,c_\delta^*)$, every solution of the closed-loop system consisting of~\eqref{eq-eta-sys},~\eqref{eq-eta-cont} with initial condition $\eta(0)\in\Omega_c$ exists for all $t\ge 0$, satisfies
\begin{equation}
\eta(t)\in\Omega_c\subset D^*,\qquad u(t)\ge 0,\qquad \forall t\ge 0,
\end{equation}
and obeys
\begin{equation}
r(\eta(t))\le \beta_c\bigl(r(\eta(0)),t\bigr)+\mu_c(\delta),\qquad \forall t\ge 0.
\end{equation}
In particular, controller~(17) renders system~(26) admissibly semi-globally practically asymptotically stable on every $\Omega_c$ with $0<c<c_\delta^*$.
\end{corollary}

\begin{proof}
The two neural-operator error bounds imply $|e_1|+|e_2|<\delta$. The claim follows directly from Theorem~\ref{thm-spas}, since \eqref{eq-eta-cont} is the approximate controller corresponding to the errors $e_1,e_2$ applied to the $\eta$-system \eqref{eq-eta-sys}.
\end{proof}

\section{Proofs of the Theorems}
\label{app:LS-Lipschitz}
\label{sec-proofs}

\subsection{Proof of Lipschitzness of Lotka-Sharpe operator}

\begin{proof}[Proof of Theorem \ref{thm-LS-Lip}]
The proof proceeds in six steps.
\paragraph*{Step 1: Well-posedness and bounds}
We consider the set of functions
\begin{align}
B = \left\{ (k,\mu) \in C^0\left([0,A];\mathbb{R}_{\geq 0}^2\right) :
\int_0^A k(a)e^{-\int_0^a \mu(s) ds} da > 1 \right\}. \label{eq:set_B}
\end{align}
For every $(k,\mu) \in B$, there exists a unique $\zeta>0$ such that $\int_0^A k(a)e^{-\zeta a - \int_0^a \mu(s) ds} da = 1$. 
Notice that $k(a)\ge 0$, $\mu(a)\ge 0$ for all $a \in [0,A]$, and
$\int_0^A k(a)e^{-\int_0^a \mu(s) ds} da > 1 $ imply that $A \|k\|_\infty > 1 \ \forall \ (k,\mu) \in B$.

\medskip
Our goal is to find an estimate of $\zeta>0$ for arbitrary $(k,\mu)\in B$.
Since $(k,\mu)\in B$, let $\zeta$ denote the unique solution of the Lotka–Sharpe equation, so that
$\int_0^A k(a)e^{-\zeta a-\int_0^a \mu(s)\,ds}\,da=1$. 
Then we get
\begin{align}
1 &\ge e^{-\zeta A} 
    \int_0^A k(a)e^{-\int_0^a \mu(s) ds} da.
\end{align}
Consequently, we obtain the estimate
\begin{align}
\zeta \ge \frac{1}{A} 
\ln \left(
\int_0^A k(a)e^{-\int_0^a \mu(s) ds} da
\right) > 0. \label{eq:zeta_estimate1}
\end{align}
Now, define $\varepsilon = \frac{1}{2 \|k\|_\infty}$, and notice that, since $A\|k\|_\infty > 1$, we get that $\varepsilon \in \Big(0, \frac{A}{2}\Big)$. 
Since $\int_0^A k(a)e^{-\zeta a - \int_0^a \mu(s) ds} da = 1$, we get
\begin{align}
&\int_0^\varepsilon k(a)e^{-\zeta a - \int_0^a \mu(s) ds} da
+ \int_\varepsilon^A k(a)e^{-\zeta a - \int_0^a \mu(s) ds} da = 1 \notag \\
\Rightarrow\;& 
\int_0^\varepsilon k(a)e^{-\int_0^a \mu(s) ds} da
+ e^{-\zeta\varepsilon} \int_\varepsilon^A k(a)e^{-\int_0^a \mu(s) ds} da \ge 1 \notag \\
\Rightarrow\;&
\int_0^\varepsilon k(a) da 
+ e^{-\zeta \varepsilon} \int_\varepsilon^A k(a)e^{-\int_0^a \mu(s) ds} da \ge 1 \notag \\
\Rightarrow\;&
\varepsilon \|k\|_\infty 
+ e^{-\zeta\varepsilon} \int_\varepsilon^A k(a)e^{-\int_0^a \mu(s) ds} da \ge 1.
\end{align}

Since $\varepsilon = \frac{1}{2 \| k \|_\infty}$, we obtain the following estimate:
\begin{align}
\zeta &\le 2 \| k \|_\infty 
\ln \left( 
2 \int_{1/(2\|k\|_\infty)}^{A} k(a) e^{-\int_0^a \mu(s)   ds}  da \right) \notag  \\
&\le 2 \| k \|_\infty 
\ln \left( 2 \int_{0}^{A} k(a) e^{-\int_0^a \mu(s)   ds}  da \right) \notag \\
&\le 2 \| k \|_\infty \ln \left( 2 A \| k \|_\infty \right).\label{eq:zeta_estimate2}
\end{align}

\paragraph*{Step 2: Monotonicity}
Let $k_{\rm min}, \mu_{\rm max}, k_{\rm max}, \mu_{\rm min}$ be given Lipschitz functions with 
$(k_{\rm min}, \mu_{\rm max}) \in B$, $(k_{\rm max}, \mu_{\rm min}) \in B$ and
\begin{align}
k_{\rm min}(a) &\le k_{\rm max}(a), \quad \text{for all } a \in [0,A] \\
\mu_{\rm min}(a) &\le \mu_{\rm max}(a), \quad \text{for all } a \in [0,A]
\end{align}
Let the unique $\zeta_{\rm min}, \zeta_{\rm max} > 0$ for which 
\begin{align}
    \int_0^A k_{\rm min}(a)e^{-\zeta_{\rm min} a - \int_0^a \mu_{\rm max}(s) ds }da =&\; 1\,, \\
     \int_0^A k_{\rm max}(a)e^{-\zeta_{\rm max} a - \int_0^a \mu_{\rm min}(s) ds }da =&\; 1\,.
\end{align}
We next prove by contradiction that
\begin{align}
\zeta_{\rm min} \le \zeta_{\rm max}.
\end{align}

Indeed, we have
\begin{align}
\nonumber \int_0^A k_{\rm min}(a)&e^{-\zeta_{\rm min} a - \int_0^a \mu_{\rm max}(s) ds} da \\ &=\; \int_0^A k_{\rm max}(a)e^{-\zeta_{\rm max} a - \int_0^a \mu_{\rm min}(s) ds} da 
\end{align}
implying
\begin{align}
&\int_0^A k_{\rm min}(a) e^{-\int_0^a \mu_{\rm max}(s) ds} \left( e^{-\zeta_{\rm min} a} - e^{-\zeta_{\rm max} a} \right)   da  \nonumber 
\\&=\; \int_0^A e^{-\zeta_{\rm max} a} \nonumber \\ & \qquad  \times \left( k_{\rm max}(a) e^{-\int_0^a \mu_{\rm min}(s) ds} - k_{\rm min}(a) e^{-\int_0^a \mu_{\rm max}(s)   ds } \right) da. \label{eq:temp1}
\end{align}
Since $k_{\rm max}(a) e^{-\int_0^a \mu_{\rm min}(s)ds} \ge k_{\rm min}(a) e^{-\int_0^a \mu_{\rm max}(s)ds}$, 
we get from \eqref{eq:temp1} that
\begin{align}
\int_0^A k_{\rm min}(a) e^{-\int_0^a \mu_{\rm max}(s)ds} \left( e^{-\zeta_{\rm min} a} - e^{-\zeta_{\rm max} a} \right) da \ge 0. \label{eq:temp2}
\end{align}
We suppose that $\zeta_{\rm min} > \zeta_{\rm max}$. Thus, we get $e^{-\zeta_{\rm min} a} - e^{-\zeta_{\rm max} a} \le 0 
\ \forall a \in [0,A]$. Since $k_{\rm min}(a) e^{-\int_0^a \mu_{\rm max}(s)ds} \ge 0 \ \forall a \in [0,A]$, we obtain that
$\int_0^A k_{\rm min}(a) e^{-\int_0^a \mu_{\rm max}(s) ds} \left( e^{-\zeta_{\rm min} a} - e^{-\zeta_{\rm max} a}\right)da \le 0$. Consequently, by virtue of \eqref{eq:temp2}, we have
$\int_0^A k_{\rm min}(a) e^{-\int_0^a \mu_{\rm max}(s) ds} \left(e^{-\zeta_{\rm min} a} - e^{-\zeta_{\rm max} a} \right)da = 0$.

The fact that a non-positive, continuous function with zero integral 
is identically equal to zero gives
$k_{\rm min}(a) e^{-\int_0^a \mu_{\rm max}(s)ds } \left( e^{-\zeta_{\rm min} a} - e^{-\zeta_{\rm max} a} \right) \equiv 0$. By continuity and since $e^{-\zeta_{\rm min} a} - e^{-\zeta_{\rm max} a} < 0 \ \forall \ a \in (0,A]$, this implies $k_{\rm min}(a) e^{-\int_0^a \mu_{\rm max}(s) ds } = 0$, which contradicts the fact that $\int_0^A k_{\rm min}(a) e^{-\int_0^a \mu_{\rm max}(s) ds } da > 1$ (recall that $(k_{\rm min},\mu_{\rm max}) \in B$).




\paragraph*{Step 3: Admissibility of $S$}
For every $(k,\mu)\in S$ we get $\int_0^A k(a) e^{-\int_0^a \mu(s)ds} da \ge 
\int_0^A k_{\rm min}(a) e^{-\int_0^a \mu_{\rm max}(s)ds} da$. Since $(k_{\rm min},\mu_{\rm max}) \in B$ we obtain from definition \eqref{eq:set_B} that 

\noindent $\int_0^A k_{\rm min}(a) e^{-\int_0^a \mu_{\rm max}(s) ds} da > 1$. 
Therefore, we get $\int_0^A k(a) e^{-\int_0^a \mu(s)ds} da > 1$ for every $(k,\mu)\in S$. Consequently, definitions \eqref{eq:set_HG}, \eqref{eq:set_S}, and \eqref{eq:set_B} imply that
\begin{align}
S \subseteq B. 
\end{align}
Then working similarly as above, we can conclude that for every $(k,\mu)\in S$ there exists a unique $\zeta\in [\zeta_{\rm min},\zeta_{\rm max}]$ such that 
$\int_0^A k(a) e^{-\zeta a - \int_0^a \mu(s) ds} da = 1$.

\paragraph*{Step 4: Fundamental identity}
We show next that the mapping
\begin{align}
P : S \to [\zeta_{\rm min},\zeta_{\rm max}]\,, 
\end{align}
that assigns for every $(k,\mu)\in S$ the unique $\zeta\in [\zeta_{\rm min},\zeta_{\rm max}]$ for which $\int_0^A k(a) e^{-\zeta a - \int_0^a \mu(s)ds}da = 1$, i.e., $P(k,\mu)=\zeta$,
is a Lipschitz mapping (in the topology of $C^0([0,A];\mathbb{R}_{\geq 0}^2)$).
Let arbitrary $(k,\mu)\in S$, $(\tilde{k},\tilde{\mu})\in S$ be given. 
Then there exist $\zeta,\tilde{\zeta}\in [\zeta_{\rm min},\zeta_{\rm max}]$ such that
\begin{align}
\int_0^A k(a) e^{-\zeta a - \int_0^a \mu(s) ds} da = \int_0^A \tilde{k}(a)e^{-\tilde{\zeta}a - \int_0^a \tilde{\mu}(s) ds} da = 1.
\end{align}
Then,
\begin{align}
&\int_0^A k(a) e^{-\zeta a - \int_0^a \mu(s)ds} da = \int_0^A \tilde{k}(a) e^{-\tilde{\zeta}a - \int_0^a \tilde{\mu}(s)ds }da = 1
\end{align}
implying
\begin{align} \int_0^A &k(a) e^{-\int_0^a \mu(s)ds} \left( e^{-\zeta a}-e^{-\tilde{\zeta}a} \right)da 
\\ &=\;  \int_0^A e^{-\tilde{\zeta}a} \left(\tilde{k}(a) e^{-\int_0^a \tilde{\mu}(s)ds}
- k(a) e^{-\int_0^a \mu(s)ds} \right) da.\label{eq:temp3} 
\end{align}

\paragraph*{Step 5: Core Lipschitz estimate}
Suppose that $\zeta \le \tilde{\zeta}$. Then, $e^{-\zeta a}- e^{-\tilde{\zeta}a} \ge 0$ for all $a\in[0,A]$ and we get that
\begin{align}
\nonumber \bigg|&\int_0^A k(a) e^{-\int_0^a \mu(s)ds} \left(e^{-\zeta a} - e^{-\tilde{\zeta}a}\right)\, da\bigg|
\\ &=\;  \int_0^A k(a)e^{-\int_0^a \mu(s)ds} \left( e^{-\zeta a}- e^{-\tilde{\zeta}a} \right) da.
\end{align}
Therefore, we get from \eqref{eq:temp3} that
\begin{align}
   \nonumber \int_0^A &k(a)e^{-\int_0^a \mu(s)ds} \left( e^{-\zeta a}- e^{-\tilde{\zeta}a} \right) da 
   \\ &\leq \int_0^A  e^{-\tilde{\zeta}a} \left| \tilde{k}(a) e^{-\int_0^a \tilde{\mu}(s) ds} - k(a) e^{-\int_0^a \mu(s) ds} \right| da 
\end{align}

Since $0 < \zeta_{\rm min} \leq \zeta \leq \tilde{\zeta} \leq \zeta_{\rm max}$, we obtain that
$ e^{-\zeta a} - e^{-\tilde{\zeta}a} \geq a(\tilde{\zeta}-\zeta)e^{-\zeta_{\rm max} a} \ \forall \ a \in [0,A]$. Thus, we obtain
\begin{align}
(\tilde{\zeta}&-\zeta) \int_0^A a e^{-\zeta_{\rm max} a} k(a) e^{-\int_0^a {\mu}(s) ds} da \notag \\
\leq& \int_0^A e^{-\tilde{\zeta}a} da \max_{r \in [0,A]}
\left( \left| \tilde{k}(r) e^{-\int_0^r \tilde{\mu}(s) ds } - k(r) e^{-\int_0^r \mu(s) ds} \right| \right) \notag \\
=&\; \frac{1 - e^{-\tilde{\zeta} A}}{\tilde{\zeta}}
\nonumber  \\ & \qquad \times \max_{r \in [0,A]} 
\left( \left| \tilde{k}(r) e^{-\int_0^r \tilde{\mu}(s) ds } - k(r) e^{-\int_0^r \mu(s) ds} \right| \right) \notag \\
\leq&\;  \frac{1 -e^{-\tilde{\zeta} A}}{\zeta_{\rm min}} \notag  \\ & \qquad \times
\max_{r \in [0,A]} 
\left( \left| \tilde{k}(r) e^{-\int_0^r \tilde{\mu}(s) ds } - k(r) e^{-\int_0^r \mu(s) ds} \right| \right)
\end{align}

By definition \eqref{eq:set_S} of the set $S$ and since $(k,\mu) \in S$, we get that
\begin{align}
\int_0^A a &e^{-\zeta_{\rm max} a} k(a) e^{-\int_0^a \mu(s) ds} da \\ \geq&\;   
e^{-\zeta_{\rm max} A} \int_0^A a k_{\rm min}(a) e^{-\int_0^a \mu_{\rm max}(s) ds} da.
\end{align}
Hence, we obtain the following estimate when $\zeta \leq \tilde{\zeta}$:
\begin{align}
\tilde{\zeta} - \zeta \leq L_1 \max_{r \in [0,A]} 
\left( \left| \tilde{k}(r) e^{-\int_0^r \tilde{\mu}(s) ds } - k(r) e^{-\int_0^r \mu(s) ds} \right| \right)
\end{align}
where
\begin{align}
L_1 := \frac{e^{\zeta_{\rm max} A} - 1}{\zeta_{\rm min} \int_0^A a k_{\rm min}(a) e^{-\int_0^a \mu_{\rm max}(s) ds } da } .\label{eq:L1}
\end{align}
Exploiting the fact that $(k_{\rm min},\mu_{\rm max})\in B$, $(k_{\rm max},\mu_{\rm min})\in B$, and estimates \eqref{eq:zeta_estimate1}, \eqref{eq:zeta_estimate2}, we obtain from \eqref{eq:L1} that $L_1 \leq L$ where 
\begin{align}
L :=\;&
\frac{
A\,\bigl(2A\|k_{\rm max}\|_\infty\bigr)^{2A\|k_{\rm max}\|_\infty-1}
}{
\left(\int_0^A a\,k_{\rm min}(a)I(a)\,da\right)
\ln\!\left(\int_0^A k_{\rm min}(a)I(a)\,da\right)
}\,,  
\end{align}
with
\begin{align}
    I(a) :=\;& e^{-\int_0^a \mu_{\rm max}(s) ds }\,.
\end{align}

Exchanging the roles of $\zeta,\tilde{\zeta} \in [\zeta_{\rm min},\zeta_{\rm max}]$, we obtain the following estimate without any assumption about $\zeta,\tilde{\zeta} \in [\zeta_{\rm min},\zeta_{\rm max}]$:
\begin{align}
\big|\tilde{\zeta}-\zeta\big|\leq L_1 \max_{r\in[0,A]} \left( \left|\tilde{k}(r)e^{-\int_0^r \tilde{\mu}(s)ds} - k(r)e^{-\int_0^r \mu(s) ds} \right| \right).
\end{align}

\paragraph*{Step 6: Reduction to sup norms}
Exploiting definition \eqref{eq:set_S} of the set $S$, and the fact that $(k,\mu)\in S$, $(\tilde{k},\tilde{\mu})\in S$, we also get
\begin{align}
\left|\tilde{\zeta}-\zeta\right| 
\leq&\; L_1 \max_{r\in[0,A]} \left( e^{-\int_0^r \tilde{\mu}(s)ds} \left|\tilde{k}(r)-k(r)\right| \right) 
\nonumber \\ &+ L_1 \max_{r\in[0,A]} \left(k(r) \left| e^{-\int_0^r \tilde{\mu}(s)ds} - e^{-\int_0^r \mu(s) ds} \right| \right) \notag \\
\leq&\; L_1 \|\tilde{k}-k\|_\infty \nonumber \\ &+ L_1 \|k_{\rm max}\|_\infty \max_{r\in[0,A]}
\left (\left| e^{-\int_0^r \tilde{\mu}(s)ds} - e^{-\int_0^r \mu(s)ds}\right| \right)
\end{align}

Since $\left|e^{-\int_0^r \tilde{\mu}(s)ds} - e^{-\int_0^r \mu(s) ds}\right| \leq \int_0^r |\mu(s)-\tilde{\mu}(s)| ds$, we get
\begin{align}
\left|\tilde{\zeta}-\zeta\right|
\leq&\; L_1 \|\tilde{k}-k\|_\infty\nonumber \\ & + L_1 \|k_{\rm max}\|_\infty \max_{r\in[0,A]} \left (\int_0^r |\mu(s)-\tilde{\mu}(s)| ds \right) \notag \\
=&\; L_1 \|\tilde{k}-k\|_\infty + L_1 \|k_{\rm max}\|_\infty \int_0^A |\mu(s)-\tilde{\mu}(s)|\, ds \notag \\
\leq&\; L_1 \|\tilde{k}-k\|_\infty + L_1 \|k_{\rm max}\|_\infty\, A\, \|\tilde{\mu}-\mu\|_\infty
\end{align}
Thus, we get for all $(k,\mu)\in S$, $(\tilde{k},\tilde{\mu})\in S$
\begin{align}
\left|\tilde{\zeta}-\zeta \right|
\leq L_1 \|\tilde{k}-k\|_\infty + L_1 \|k_{\rm max}\|_\infty A \|\tilde{\mu}-\mu\|_\infty,
\end{align}
where $L_1$ is given by \eqref{eq:L1}.
\end{proof}

\subsection{Proof of stability under Lotka-Sharpe parameter errors}

\begin{proof}[Proof of Theorem \ref{thm-spas}]
The perturbed closed system is
\begin{eqnarray}
\dot\eta_1 &=& \dot\eta_1^{\mathrm{nom}}(\eta)-\Delta_u(\eta,e_1,e_2), \\
\dot\eta_2 &=& \dot\eta_2^{\mathrm{nom}}(\eta)-\Delta_u(\eta,e_1,e_2),
\end{eqnarray}
where
\begin{eqnarray}
\dot\eta_1^{\mathrm{nom}}(\eta)
&=&
-\frac{\beta}{x_1^*(0)\gamma_2}(1-e^{-\eta_1})
-
\bigl(1+\beta(1+\varepsilon)\bigr) \nonumber \\ &&\times 
\left(\zeta_1-\zeta_2+\frac{1}{x_1^*(0)\gamma_2}\right)(e^{\eta_2}-1), \\
\dot\eta_2^{\mathrm{nom}}(\eta)
&=&
\frac{1-\beta}{x_1^*(0)\gamma_2}(1-e^{-\eta_1})
-
\beta(1+\varepsilon) \nonumber \\ && \times 
\left(\zeta_1-\zeta_2+\frac{1}{x_1^*(0)\gamma_2}\right)(e^{\eta_2}-1).
\end{eqnarray}
and
\begin{equation}
\Delta_u(\eta,e_1,e_2)=-\,e_2-(\hat a-a)+\beta\,\Delta_{\mathrm{gain}}(\eta,e_1,e_2).
\end{equation}
\begin{eqnarray}
\Delta_{\mathrm{gain}}(\eta,e_1,e_2)&:=& \hat S(\eta)-S(\eta)
\,,
\end{eqnarray}
where functions $S(\cdot), \hat S(\cdot)$ and constant $\hat a$ are defined in Appendix \ref{app-functions+constants}. 

Next, we have that
\begin{equation}
\dot V_1
=
-\begin{bmatrix}\phi_1 & \phi_2\end{bmatrix}
Q
\begin{bmatrix}\phi_1\\ \phi_2\end{bmatrix}
-
d(\eta)
\Delta_u(\eta,e_1,e_2),
\end{equation}
where
\begin{equation}
d(\eta):=\phi_1(\eta_1)+(1+\varepsilon)\phi_2(\eta_2),
\end{equation}
We first verify that $c^*_\delta > 0$. At $\eta = 0$ and $e_1 = e_2 = 0$, one has $u(0;0,0) = u^* > 0$ by (10), and $V_1(0) = 0$. By continuity of $u$ in $(\eta,e_1,e_2)$ and of $V_1$ in $\eta$, both conditions defining $c^*_\delta$ in (51) hold for all sufficiently small $c$ and $\delta$, so the supremum in (51) is taken over a non-empty set and $c^*_\delta > 0$. Since $c<c_\delta^*$, the set $\Omega_c$ is compact and contained in $\mathcal D_*$. Since $r(\eta)$ is continuous, it attains its maximum on $\Omega_c$. Define
\begin{equation}
R(c):=\max_{\eta\in\Omega_c} r(\eta).
\end{equation}
Then $R(c)\in(0,\min\{a,b\})$ and
\begin{equation}
\Omega_c\subset \mathcal T_{R(c)}:=\{\eta:\ r(\eta)\le R(c)\}.
\end{equation}
Hence, by Lemma \ref{lem-Delta_u}, there exists $C_{R(c)}>0$ such that
\begin{align}
|\Delta_u(\eta,e_1,e_2)|\le& C_{R(c)}(|e_1|+|e_2|),\nonumber \\ &\qquad  \forall\,\eta\in\Omega_c,\quad |e_1|+|e_2|\le\delta.
\end{align}
Additionally,
\begin{equation}
|d(\eta)|
\le \sqrt{1+(1+\varepsilon)^2}\,r
=:c_\varepsilon r,
\qquad
c_\varepsilon:=\sqrt{1+(1+\varepsilon)^2}, 
\end{equation}
and
\begin{equation}
\begin{bmatrix}\phi_1 & \phi_2\end{bmatrix}
Q
\begin{bmatrix}\phi_1\\ \phi_2\end{bmatrix}
\ge \lambda_*(\phi_1^2+\phi_2^2)=\lambda_* r^2,
\end{equation}
where
\begin{equation}
\lambda_*:=\lambda_{\min}(Q)>0.
\end{equation}
Hence
\begin{equation}
\dot V_1\le -\lambda_* r^2+c_\varepsilon |\Delta_u|\,r.
\end{equation}
By Young’s inequality,
\begin{equation}
c_\varepsilon |\Delta_u|\,r
\le \frac{\lambda_*}{2}r^2+\frac{c_\varepsilon^2}{2\lambda_*}\Delta_u^2,
\end{equation}
so
\begin{equation}
\dot V_1
\le
-\frac{\lambda_*}{2}(\phi_1^2+\phi_2^2)
+\frac{c_\varepsilon^2}{2\lambda_*}\bar\Delta^2,
\end{equation}
where
\begin{equation}
\bar\Delta := \sup_{\eta\in\mathcal T_{R(c)}}|\Delta_u(\eta,e_1,e_2)|\,,
\end{equation}
and therefore
\begin{equation}
\dot V_1<0
\qquad\text{whenever}\qquad
\phi_1^2+\phi_2^2>\frac{c_\varepsilon^2}{\lambda_*^2}\bar\Delta^2,
\end{equation}
so $\Omega_c$ is forward-invariant. 
The explicit expressions for $\beta_c, \mu_c$ are obtained as
\begin{eqnarray}
\beta_c(s,t)&:=&\sqrt{\frac{M_c}{m_c}}\,e^{-\frac{\lambda_*}{4M_c}t}\,s
\\
\mu_c(\delta)&:=&\frac{c_\varepsilon}{\lambda_*}\sqrt{\frac{M_c}{m_c}}\,C_{R(c)}\,\delta
\end{eqnarray}
where 
\begin{eqnarray}
m_c &:=&\frac{1}{2}\min\!\left\{\frac{e^{-3B_1(c)}}{a},\frac{(1+\varepsilon)e^{-3B_2(c)}}{b}\right\}
\\
M_c&:=&\frac{1}{2}\max\!\left\{\frac{e^{3B_1(c)}}{a},\frac{(1+\varepsilon)e^{3B_2(c)}}{b}\right\}
\\
B_1(c)&:=&1+\frac{c}{a}
\\
B_2(c)&:=&1+\frac{c}{(1+\varepsilon)b}
\end{eqnarray}
Furthermore, one can simplify, by majorization, 
\begin{equation}
\frac{M_c}{m_c}
\le
e^{3(B_1(c)+B_2(c))}
\left(
q+\frac{1}{q}
\right)
\end{equation}
where
\begin{eqnarray}
q &:=& 
(1+\varepsilon) q_0e^{-3c\,x_1^*(0)\gamma_2\left(1-\frac{q_0}{(1+\varepsilon)}\right)}
\\
q_0&:=&
\frac{1}{1+(\zeta_1-\zeta_2)x_1^*(0)\gamma_2}
\end{eqnarray}
\end{proof}

\section{Numerical Results}
\label{sec-sims}


\begin{figure*}[t]
    \centering
    \includegraphics[width=\linewidth]{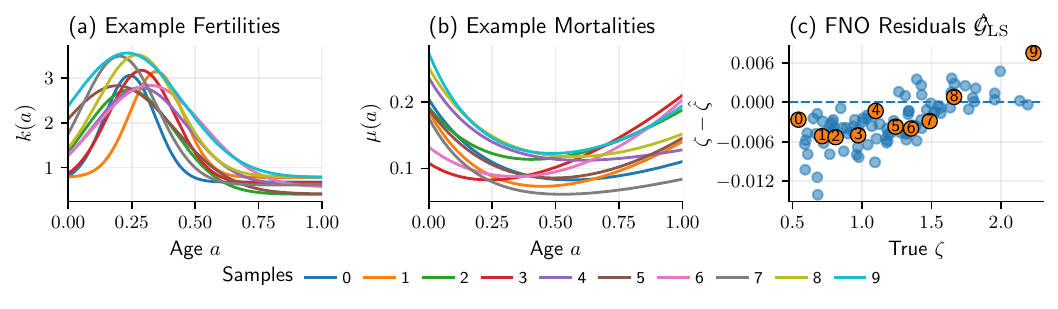}
    \caption{Example of various $(k, \mu)$ used in training and performance of learned operator $\widehat{\mathcal{G}}_{\rm LS}$. (c) highlights the residuals of all $100$ test examples during training in blue and the samples corresponding to (a) and (b) in orange.}
    \label{fig:fno}
\end{figure*}

\begin{figure*}[t]
    \centering
    \includegraphics[width=\linewidth]{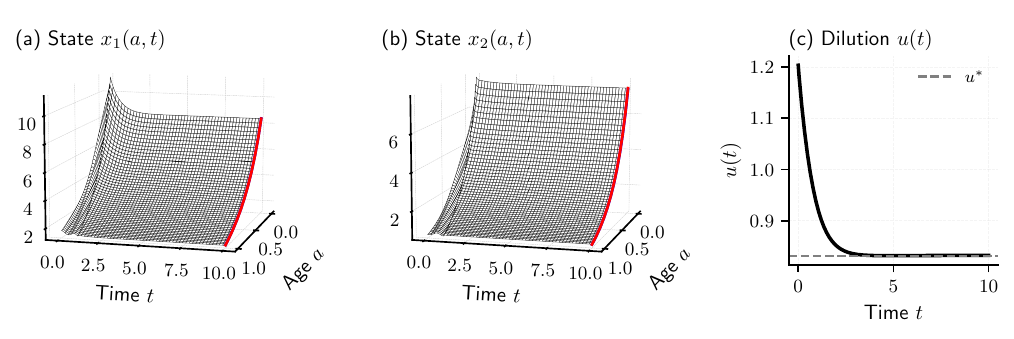}
    \caption{Simulation profiles of the (a) population $x_1$, (b) population $x_2(t)$, and (c) dilution control using $(\hat{\zeta}_1, \hat{\zeta_2}) = \widehat{\mathcal{G}}_{\rm{LS}} (k, \mu) $ in the control law \eqref{eq-eta-cont}. $k_1, k_2$ and $\mu_1, \mu_2$ correspond with samples eight and nine in Figure \ref{fig:fno} We consider the initial condition case where there are a large amount of species $x_1$ compared to species $x_2$ in the beginning, but become similarly concentrated due dilution choice $u^\ast = 0.83$ ($x_1^\ast(0) = 8.45, x_2^\ast(0) = 7.42$) is to obtain a larger species $x_1$. } 
    \label{fig:simple-sim}
\end{figure*}

All code, datasets and models are publicly available in 
\cite{krstic2026neuraloperatorsagestructured}. 
\subsection{Learning $\mathcal{G}_{\mathrm{LS}}$ operator}
To learn $\mathcal{G}_{\mathrm{LS}}$, we first need to construct a dataset of biologically relevant $k, \mu, g$ profiles. Consider the family of functions:
\begin{align*}
k(a)=&\;k_{\mathrm{base}}+k_{\mathrm{amp}}\exp\!\left(-\frac{(a-k_{\mathrm{center}})^2}{2k_{\sigma}^2}\right),\\
\mu(a)=&\;\mu_{\rm base}+\mu_{\mathrm{juv, amp}}e^{-\mu_{\mathrm{juv}}a}+\mu_{\mathrm{sen, amp}}a^{\mu_{\mathrm{sen}}},\\
g(a)=&g_{\mathrm{base}}+g_{\mathrm{amp}}\exp\!\left(-\frac{(a-g_{\mathrm{center}})^2}{2g_{\sigma}^2}\right).
\end{align*}
We choose a Gaussian fertility profile $k$, centered at early adult ages, to reflect the biological fact that fertility typically becomes viable only after maturation and is concentrated within a finite reproductive window \cite{Delbaere2020-gf}. Likewise, the mortality profile 
$\mu$ is designed to reproduce the broadly observed bathtub-shaped age pattern, with elevated mortality at the beginning and end of life and lower mortality during prime ages \cite{Chu2008-ix}. Finally, the interaction profile $g$ is taken to peak in mature individuals, reflecting the idea that ecologically important interactions such as hunting or complex foraging are often strongest when individuals have reached adult body condition and accumulated sufficient experience and skill.

To generate variability across families, we sample the parameters independently from uniform distributions over prescribed ranges. Specifically,
\begin{align*}
k_{\mathrm{base}}   &\sim \mathrm{Unif}(0.40,\,0.80) &\qquad
k_{\mathrm{amp}}    &\sim \mathrm{Unif}(2.00,\,3.00) \\
k_{\mathrm{center}} &\sim \mathrm{Unif}(0.11,\,0.35) &\qquad
k_{\sigma}          &\sim \mathrm{Unif}(0.05,\,0.23) \\
\mu_{\mathrm{base}}          &\sim \mathrm{Unif}(0.03,\,0.10) &\qquad
\mu_{\mathrm{juv, amp}}  &\sim \mathrm{Unif}(0.05,\,0.19) \\
\mu_{\mathrm{juv}}&\sim \mathrm{Unif}(3.5,\,5.5)   &\qquad
\mu_{\mathrm{sen, amp}}  &\sim \mathrm{Unif}(0.03,\,0.17) \\
\mu_{\mathrm{sen}}&\sim \mathrm{Unif}(1.7,\,2.9)   &\qquad
g_{\mathrm{base}} &\sim \mathrm{Unif}(0.05,\,0.13) \\
g_{\mathrm{amp}}    &\sim \mathrm{Unif}(0.20,\,0.50) &\qquad
g_{\mathrm{center}} &\sim \mathrm{Unif}(0.37,\,0.63) \\
g_{\sigma}          &\sim \mathrm{Unif}(0.05,\,0.31) &
\end{align*}
For each sample \((K,\mu)\), we compute the net reproduction number $R_0=\int_0^1 k(a)\,\Pi(a)\,da$, 
and retain only those samples satisfying $R_0 > 1.2$ such that we ensure $\zeta_1, \zeta_2 > 0$ enabling positive dilution \eqref{eq:u_star_constraint}. 

To learn the mapping $\mathcal{G}_{\rm LS}$, we sample $1000$ different pairs of $((k, \mu), \zeta)$ where $\zeta$ is identifies via a numerical precision root-finding algorithm. We then train a Fourier Neural Operator (FNO) \cite{li2021Fourier} consisting of $4$ layers with $16$ Fourier modes and a hidden size of $64$ neurons per layer. We use a learning rate of $4 \times 10^{-3}$ with the AdamW optimizer \cite{loshchilov2018decoupled} achieving a training mean squared error of $3.4 \times 10^{-5}$ after $100$ epochs. We present an example of the training functions $(k, \mu)$ as well as the performance of the FNO qualitatively in \ref{fig:fno}. Notice that, across magnitudes of $\zeta$, the error of the FNO residuals remains concentrated between $\pm 0.001$ indicating very strong performance in learning $\widehat{\mathcal{G}}_{\rm LS}$. 

In Figure \ref{fig:simple-sim}, we validate Corollary \ref{cor-main} by simulating the closed-loop system with $(\hat{\zeta}_1,\hat{\zeta}_2)$ obtained from the learned operator $\widehat{\mathcal{G}}_{\rm LS}$. We choose an initial condition in which the prey population dominates and regulate the system toward a target equilibrium where both species have similar concentrations. The trajectories converge to the prescribed equilibrium and the control remains positive for all simulated times, in agreement with \eqref{eq:u_star_constraint}. This numerical example illustrates that the neural approximation of $\mathcal{G}_{\rm LS}$ preserves the stabilizing behavior predicted by the theory and supports the practical use of operator learning in the feedback design.

\section{Adaptive example when mortality and fertility are unknown} \label{sec-adaptive}
In Section \ref{sec-sims}, the numerical implementation relied on a one-time learning of $\zeta_1$ and $\zeta_2$ under the assumption that the kernels $(k,\mu)$ were known. In practice, however, exact knowledge of the fertility and mortality profiles may be unavailable, making a purely offline deployment of $\widehat{\mathcal{G}}_{\rm LS}$ intractable. As a final remark, to showcase the value in learning $\widehat{\mathcal{G}}_{\rm LS}$ with a neural operator, we present a simple adaptive design together with an illustrative simulation, for the case in which the learned operator must be re-evaluated online.

In this section we do not pursue a theoretical study of stability under adaptive update laws. This would be possible but would require a large page budget, which bringing little additional insight relative to \cite{lamarque2025Adaptive}. 

\subsection{Adaptive design} \label{subsec-adaptive-design}
Among the uncertain quantities, the boundary kernel $k$ is the simplest to adapt. Indeed, the boundary condition for both state components is of the form
\begin{align}
x_i(t,0)= \int_0^a k_i(\alpha, t) x_i(\alpha, t)d \alpha\,, \qquad i \in \{1, 2\}\,,
\end{align}
which is linear in $k$ and involves no time or age derivatives of the parameter. This suggests the direct gradient-type update law
\begin{align} 
\frac{\partial \hat{k}_i(t,a)}{\partial t}
=\;& \nonumber 
\Gamma_{k, i}\,  \frac{x_i(t,a)}{1 + \int_0^A x_i(\alpha, t)^2d\alpha}\Bigg(x_i(t,0) \\ &-  \int_0^A \hat{k}_i(\alpha, t) x_i(\alpha, t) d\alpha\Big), \label{eq:khat-update}
\end{align}
where $\Gamma_{k, i} > 0$ is a gain parameter.
A similar modular construction can be used to estimate the age-dependent mortality profile $\mu_i$. For each $i\in\{1,2\}$, rewrite the population dynamics as
\begin{align}
\label{eq-mu-parametrization}
\partial_t x_i(a,t)=r_i(a,t)-\mu_i(a)x_i(a,t),
\end{align}
where $r_i(a,t)$ collects the known transport and interaction terms. In particular,
\begin{align}
r_1(a,t)=\;&\;-\partial_a x_1(a,t)-u(t)\nonumber \\ &-\int_0^A g_1(\alpha)x_2(\alpha,t)\,d\alpha x_1(a,t),\\
r_2(a,t)=\;&-\partial_a x_2(a,t)-u(t) \nonumber \\ &-\frac{1}{\int_0^A g_2(\alpha)x_1(\alpha,t)\,d\alpha}.
\end{align}
For the clarity of presentation, we treat the functions $g_1(a), g_2(a)$ as known. The estimation of $g_1$ is easy, whereas the estimation of $g_2$ requires overparametrization. Since our goal here is pedagogical --- to elucidate the applicability of the neural opeartor $\mathcal{G}{\rm LS}$ in adaptive control, we proceed with only $k$ and $\mu $ treated as unknown.

For each fixed age $a$, \eqref{eq-mu-parametrization} is a scalar linear equation in time with unknown parameter $\mu_i(a)$. Introducing the filters
\begin{align}
\partial_t \sigma_i(a,t)&=-\alpha_i \sigma_i(a,t)+x_i(a,t),\\
\partial_t \rho_i(a,t)&=-\alpha_i \rho_i(a,t)+r_i(a,t),
\end{align}
with $\alpha_i>0$, one obtains the pointwise regression
\begin{align}
Y_i(a,t)=\mu_i(a)\sigma_i(a,t),
\end{align}
where
\begin{align}
Y_i(a,t)=\rho_i(a,t)-x_i(a,t)+\alpha_i\sigma_i(a,t)+e^{-\alpha_i t}x_{i,0}(a).
\end{align}
Neglecting the exponentially decaying transient, this suggests the gradient update
\begin{align}
\partial_t \hat{\mu}_i(a,t)
=
\Gamma_{\mu,i}\,
\frac{\sigma_i(a,t)}{1+\sigma_i(a,t)^2}
\Big(
Y_i(a,t)-\hat{\mu}_i(a,t)\sigma_i(a,t)
\Big), \label{eq:uhat-update}
\end{align}
with $\Gamma_{\mu,i}>0$ is the gain parameter.

The main advantage of this design is that this preserves the full age dependence of $\mu_i$; the trade-off is that it requires the spatial derivative $\partial_a x_i$, which will require finite difference estimation in practice as it is typically not measurable. 

\begin{figure*}[t]
    \includegraphics[width=\textwidth]{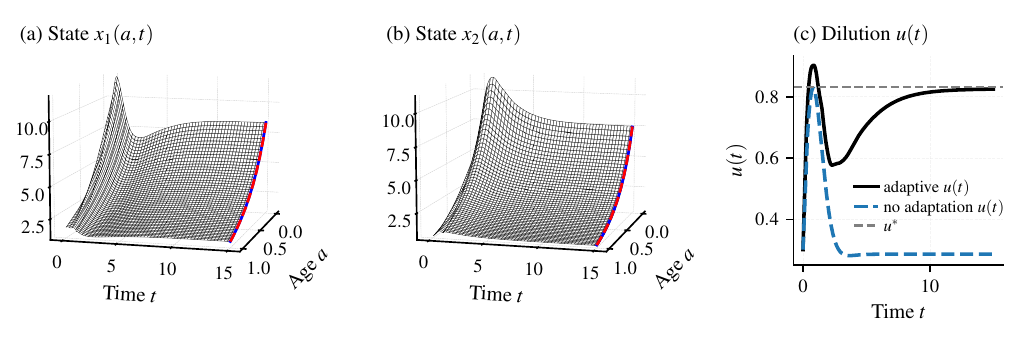}
    \includegraphics[width=\textwidth]{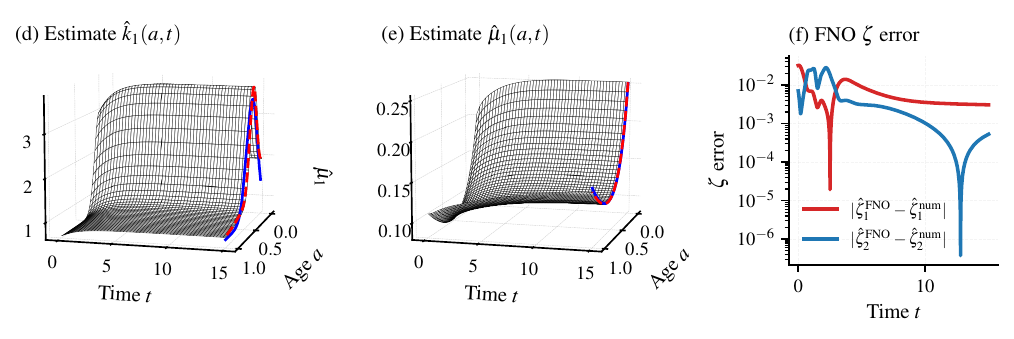}
    \caption{Adaptive control simulation where $k, \mu, g$ are   samples eight and nine from Figure \ref{fig:fno} and $\hat{k}_i, \hat{\mu}_i$ are initialized with samples zero and one from Figure \ref{fig:fno}. The blue lines in (a), (b), (d), (e) are the final states achieved and the red dashed lines are the target equilibrium and true $k, \mu$ functions for (a), (b) and (d), (e) respectively. All the initial conditions and the dilution point are the same as in Figure \ref{fig:simple-sim} and the controller uses the adaptive update laws in Section \ref{subsec-adaptive-design} along with $\widehat{\mathcal{G}}_{\rm LS}(\hat{k}_i(\cdot, t), \hat{\mu}_i(\cdot, t)) = \hat{\zeta}_i(t)$ to compute the control law $u(t)$.}
    \label{fig:adaptive}
\end{figure*}

\subsection{Illustrative simulation}
To construct a numerical dataset for $\widehat{\mathcal{G}}_{\rm LS}$, it is not sufficient to use the $(k,\mu)$ samples from Section~\ref{sec-sims}, since the adaptive controller may produce estimates $\hat{k}$ and $\hat{\mu}$ with substantially different structure. We therefore generate a new dataset by simulating $100$ sampled triples $(k,\mu,g)$ as in Section~\ref{sec-sims}  where $g$ is known and $k, \mu$ are unknown and hence updated with adaptive estimators. Further, at each time step, we use a machine-precision root-finding solver to compute $\hat{\zeta}$ from the adaptive estimates $(\hat{k}, \hat{\mu})$ to obtain $u(t)$. From each trajectory, we sample $200$ random time points, yielding a diverse training set of $20{,}000$ pairs $(\hat{k},\hat{\mu})$. We train for $100$ epochs using the same FNO architecture and training settings as in Section~\ref{sec-sims}, achieving a test error of $2\times 10^{-4}$.

Figure~\ref{fig:adaptive} shows a representative adaptive closed-loop simulation under the same biological functions $(k_i, \mu_i, g_i)$ as in Figure~\ref{fig:simple-sim}, but with mismatched initial estimates $\hat{k}$ and $\hat{\mu}$. At each time $t$, the control input is computed from $\hat{\zeta}_i=\widehat{\mathcal{G}}_{\rm LS}(\hat{k}_i,\hat{\mu}_i)$, where $\hat{k}_i$ and $\hat{\mu}_i$ are the current adaptive estimates generated from \eqref{eq:khat-update} and \eqref{eq:uhat-update} respectively. Despite the incorrect initialization, the learned operator combined with adaptation drives the system to the target dilution and the desired predator--prey equilibrium. This illustrates that the learned FNO remains effective when driven by online parameter estimates rather than the true profiles, which is the practically relevant setting since mortality and fertility rates are typically not known exactly.

\section{Conclusions}
In this work, we gave the first stabilizing feedback design for an age-structured predator–prey system with an approximated Lotka–Sharpe parameter 
$\zeta$. Specifically, we established Lipschitz continuity of the implicitly defined Lotka–Sharpe operator on a biologically admissible domain, which in turn guarantees the existence of uniformly accurate neural-operator approximations. We then quantified how approximation errors in $\zeta$ propagate through the three operators defining the feedback law and used these bounds to derive a robust stability result guaranteeing semi-global practical asymptotic stability under a positivity constraint on the control input. Overall, these results provide the first mathematically rigorous foundation for stabilizing biologically relevant age-structured predator–prey systems by learning-based feedback in which the feedback law necessarily depends on an approximate Lotka–Sharpe operator.

\appendix

\section{Appendices}

\subsection{Functions and constants used in Theorem \ref{thm-spas}}
\label{app-functions+constants}

\begin{equation}
S(\eta)=(1+\varepsilon)(\zeta_2-\zeta_1)-\varepsilon a-m_1e^{-\eta_1}+(1+\varepsilon)m_2e^{\eta_2}
\end{equation}
\begin{align}
\hat S(\eta)-S(\eta)
=&\;
(1+\varepsilon)(e_1-e_2)-\varepsilon(\hat a-a)\nonumber \\ & -(\hat m_1-m_1)e^{-\eta_1}+(1+\varepsilon)(\hat m_2-m_2)e^{\eta_2}
\end{align}
\begin{equation}
\hat a:=\frac{1}{x_1^*(0)\bigl(\gamma_2+\Gamma_2(e_1)\bigr)}
\end{equation}
\begin{align}
m_1:=&\;\frac{\kappa_1}{x_1^*(0)\gamma_2\langle \pi_{0,1},n_1\rangle},
\\ 
\hat m_1:=&\;\frac{\kappa_1+K_1(e_1)}
{x_1^*(0)\bigl(\gamma_2+\Gamma_2(e_1)\bigr)\bigl(\langle \pi_{0,1},n_1\rangle+P_1(e_1)\bigr)} \\ 
m_2:=&\;\frac{\gamma_1x_2^*(0)\langle \pi_{0,2},n_2\rangle}{\kappa_2},\\
\hat m_2:=&\;\frac{\bigl(\gamma_1+\Gamma_1(e_2)\bigr)x_2^*(0)\bigl(\langle \pi_{0,2},n_2\rangle+P_2(e_2)\bigr)}
{\kappa_2+K_2(e_2)} \\ 
\Gamma_1(e_2):=&\;\mathcal{G}_\gamma(g_1,\zeta_2-e_2,\mu_2)-\gamma_1\\
\Gamma_2(e_1):=&\;\mathcal{G}_\gamma(g_2,\zeta_1-e_1,\mu_1)-\gamma_2 \\ 
K_1(e_1):=&\;\mathcal{G}_\kappa(k_1,\mu_1,\zeta_1-e_1)-\kappa_1 \\ 
K_2(e_2):=&\;\mathcal{G}_\kappa(k_2,\mu_2,\zeta_2-e_2)-\kappa_2 \\ 
P_1(e_1):=&\; \left\langle \mathcal{G}_\pi(k_1,\mu_1,\zeta_1-e_1)-\pi_{0,1},\,n_1\right\rangle, \\ 
P_2(e_2):=&\; \left\langle \mathcal{G}_\pi(k_2,\mu_2,\zeta_2-e_2)-\pi_{0,2},\,n_2\right\rangle
\end{align}


\subsection{Lipschitzness of operators other than the Lotka-Sharpe}
\label{app-Lip-other}

\begin{lemma}[Lipschitz continuity of $\mathcal{G}_\gamma$, $\mathcal{G}_\kappa$, $\mathcal{G}_\pi$ in $\zeta$]
\label{lem-Lip-other}
Let $g, k \in C^0([0,A];\mathbb{R}_{\geq 0})$, $\mu \in C^0([0,A];\mathbb{R}_{\geq 0})$, and let $[\zeta_{\min},\zeta_{\max}] \subset \mathbb{R}_{\geq 0}$ be a compact interval. Then:

\smallskip
\noindent\textit{(i)} The map $\zeta \mapsto \mathcal{G}_\gamma(g,\zeta,\mu)$ is Lipschitz on $[\zeta_{\min},\zeta_{\max}]$ with constant
\begin{equation}
L_\gamma := A\|g\|_\infty e^{\zeta_{\max} A}.
\end{equation}

\noindent\textit{(ii)} The map $\zeta \mapsto \mathcal{G}_\kappa(k,\mu,\zeta)$ is Lipschitz on $[\zeta_{\min},\zeta_{\max}]$ with constant
\begin{equation}
L_\kappa := A^2\|k\|_\infty e^{\zeta_{\max} A}.
\end{equation}

\noindent\textit{(iii)} The map $\zeta \mapsto \mathcal{G}_\pi(k,\mu,\zeta)$ is Lipschitz on $[\zeta_{\min},\zeta_{\max}]$, in the sup norm, with constant
\begin{equation}
L_\pi := \frac{A^2}{2}\|k\|_\infty e^{\zeta_{\max} A}.
\end{equation}
\end{lemma}

\begin{proof}
\textbf{(i)} By definition,
\begin{equation}
\mathcal{G}_\gamma(g,\zeta,\mu) = \int_0^A g(a)\,e^{-\int_0^a(\zeta+\mu(s))ds}\,da.
\end{equation}
Differentiating under the integral with respect to $\zeta$,
\begin{equation}
\frac{\partial}{\partial\zeta}\mathcal{G}_\gamma(g,\zeta,\mu) = -\int_0^A a\,g(a)\,e^{-\int_0^a(\zeta+\mu(s))ds}\,da.
\end{equation}
Taking absolute values and bounding $e^{-\zeta a} \leq e^{\zeta_{\max}A}$ and $a \leq A$,
\begin{equation}
\left|\frac{\partial}{\partial\zeta}\mathcal{G}_\gamma\right| \leq A\|g\|_\infty e^{\zeta_{\max}A} = L_\gamma.
\end{equation}
By the mean value theorem, for any $\zeta,\tilde\zeta \in [\zeta_{\min},\zeta_{\max}]$,
\begin{equation}
|\mathcal{G}_\gamma(g,\zeta,\mu) - \mathcal{G}_\gamma(g,\tilde\zeta,\mu)| \leq L_\gamma|\zeta-\tilde\zeta|.
\end{equation}

\textbf{(ii)} By definition,
\begin{equation}
\mathcal{G}_\kappa(k,\mu,\zeta) = \int_0^A a\,k(a)\,e^{-\int_0^a(\zeta+\mu(s))ds}\,da.
\end{equation}
Differentiating under the integral,
\begin{equation}
\frac{\partial}{\partial\zeta}\mathcal{G}_\kappa(k,\mu,\zeta) = -\int_0^A a^2\,k(a)\,e^{-\int_0^a(\zeta+\mu(s))ds}\,da.
\end{equation}
Bounding $a^2 \leq A^2$ and $e^{-\zeta a} \leq e^{\zeta_{\max}A}$,
\begin{equation}
\left|\frac{\partial}{\partial\zeta}\mathcal{G}_\kappa\right| \leq A^2\|k\|_\infty e^{\zeta_{\max}A} = L_\kappa.
\end{equation}
The mean value theorem gives the result.

\textbf{(iii)} By definition,
\begin{equation}
\mathcal{G}_\pi(k,\mu,\zeta)(a) = \int_a^A k(s)\,e^{\int_a^s(\zeta+\mu(l))dl}\,ds.
\end{equation}
Differentiating under the integral with respect to $\zeta$,
\begin{equation}
\frac{\partial}{\partial\zeta}\mathcal{G}_\pi(k,\mu,\zeta)(a) = \int_a^A (s-a)\,k(s)\,e^{\int_a^s(\zeta+\mu(l))dl}\,ds.
\end{equation}
Since $\int_a^A(s-a)\,ds = \frac{(A-a)^2}{2} \leq \frac{A^2}{2}$, bounding $k(s) \leq \|k\|_\infty$ and $e^{\zeta(s-a)} \leq e^{\zeta_{\max}A}$,
\begin{equation}
\left|\frac{\partial}{\partial\zeta}\mathcal{G}_\pi(k,\mu,\zeta)(a)\right| \leq \frac{A^2}{2}\|k\|_\infty e^{\zeta_{\max}A} = L_\pi.
\end{equation}
Since this bound is uniform in $a \in [0,A]$, the mean value theorem gives
\begin{equation}
\|\mathcal{G}_\pi(k,\mu,\zeta) - \mathcal{G}_\pi(k,\mu,\tilde\zeta)\|_\infty \leq L_\pi|\zeta-\tilde\zeta|. \qquad \QED
\end{equation}
\end{proof}

\subsection{Bound on control perturbation in terms of Lotka-Sharpe error}
\label{app-bound}

\begin{lemma}[$\Delta_u$ bounded in terms of $|e_1|+|e_2|$]
\label{lem-Delta_u}
Fix $\delta>0$. For every $R\in(0,\min\{a,b\})$ there exists $C_{R}>0$ such that
\begin{equation}
|\Delta_u(\eta,e_1,e_2)|\le C_{R}(|e_1|+|e_2|)
\end{equation}
for all $\eta\in\mathcal T_{R}:=\{\eta:\ r(\eta)\le R\}$ and $ |e_1|+|e_2|\le\delta$.
\end{lemma}

\begin{proof}
Fix $\delta>0$ and $R\in(0,\min\{a,b\})$. Let $\eta\in\mathcal T_{R}=\{\eta:\,r(\eta)\le R\}$. Since $r(\eta)=\sqrt{\phi_1(\eta_1)^2+\phi_2(\eta_2)^2}$, with $\phi_1(\eta_1)=a(1-e^{-\eta_1})$ and $\phi_2(\eta_2)=b(e^{\eta_2}-1)$, and since $R<\min\{a,b\}$, it follows that
\begin{equation}
1-\frac{R}{a}\le e^{-\eta_1}\le 1+\frac{R}{a},\qquad
1-\frac{R}{b}\le e^{\eta_2}\le 1+\frac{R}{b}.
\end{equation}
Hence there exist constants $E_{1,R},E_{2,R}>0$ such that
\begin{equation}
e^{-\eta_1}\le E_{1,R},\qquad e^{\eta_2}\le E_{2,R},\qquad \forall\,\eta\in\mathcal T_{R}.
\end{equation}
Let $|e_1|+|e_2|\le\delta$. Since $\hat\zeta_i=\zeta_i-e_i$, the arguments $\zeta_1-e_1$ and $\zeta_2-e_2$ lie in compact intervals, and the mappings $\mathcal{G}_\gamma$, $\mathcal{G}_\kappa$, and $\mathcal{G}_\pi$ are Lipschitz in $\zeta$ on these intervals. Hence there exist constants $L_{\Gamma_i,\delta},L_{K_i,\delta},L_{P_i,\delta}>0$ such that
\begin{align}
|\Gamma_1(e_2)|\le&\; L_{\Gamma_1,\delta}|e_2|,\\
|\Gamma_2(e_1)|\le&\; L_{\Gamma_2,\delta}|e_1|,\\
|K_1(e_1)|\le&\; L_{K_1,\delta}|e_1|, \\ 
|K_2(e_2)|\le&\; L_{K_2,\delta}|e_2|, \\ 
|P_1(e_1)|\le&\; L_{P_1,\delta}|e_1|, \\ 
|P_2(e_2)|\le&\; L_{P_2,\delta}|e_2|.
\end{align}
Since the interval $[\zeta_1-\delta,\zeta_1+\delta]$ is compact and the map 
$\zeta \mapsto \mathcal{G}_\gamma(g_2,\zeta,\mu_1)$ is continuous and strictly positive, it attains a positive minimum on this interval. 
Hence there exists $\underline{\gamma}_2>0$ such that
\begin{align}
\label{eq-lower-gamma}
\gamma_2 + \Gamma_2(e_1) =\;& \nonumber  \mathcal{G}_\gamma(g_2,\zeta_1 - e_1,\mu_1) \\\ge\;&  \min_{\zeta \in [\zeta_1-\delta,\zeta_1+\delta]} \mathcal{G}_\gamma(g_2,\zeta,\mu_1) =: \underline{\gamma}_2 > 0.
\end{align}
Similarly, there exist $\underline{\kappa}_2>0$ and $\underline{\pi}_1>0$ such that
\begin{equation}
\label{eq-lower-bnd-kappa-pi}
\kappa_2+K_2(e_2)\ge \underline{\kappa}_2>0,\qquad
\langle \pi_{0,1},n_1\rangle+P_1(e_1)\ge \underline{\pi}_1>0.
\end{equation}
By Lemma \ref{lem-Lip-other} and the lower bound in \eqref{eq-lower-gamma}, the map $\zeta \mapsto \frac{1}{x_1^*(0)\mathcal{G}_\gamma(g_2,\zeta,\mu_1)}$ is Lipschitz on $[\zeta_1-\delta,\zeta_1+\delta]$, and hence there exists $L_{a,\delta}>0$ such that
\begin{equation}
|\hat a-a|\le L_{a,\delta}|e_1|.
\end{equation}
Using the lower bounds in \eqref{eq-lower-gamma}, \eqref{eq-lower-bnd-kappa-pi}and the Lipschitz continuity of $\mathcal{G}_\gamma$, $\mathcal{G}_\kappa$, and $\mathcal{G}_\pi$, it follows that there exist constants $L_{m_1,\delta},L_{m_2,\delta}>0$ such that
\begin{equation}
|\hat m_1-m_1|\le L_{m_1,\delta}|e_1|,\qquad
|\hat m_2-m_2|\le L_{m_2,\delta}|e_2|.
\end{equation}
By definition,
\begin{equation}
\Delta_{\mathrm{gain}}(\eta,e_1,e_2):= \hat S(\eta)-S(\eta),
\end{equation}
where
\begin{align}
\hat S(\eta)-S(\eta)=&\;(1+\varepsilon)(e_1-e_2)-\varepsilon(\hat a-a)\nonumber \\ &-(\hat m_1-m_1)e^{-\eta_1} +(1+\varepsilon)(\hat m_2-m_2)e^{\eta_2}.
\end{align}
Therefore, for all $\eta\in\mathcal T_{R}$,
\begin{align}
|\Delta_{\mathrm{gain}}(\eta,e_1,e_2)|
\le&\;
(1+\varepsilon)(|e_1|+|e_2|)+\varepsilon|\hat a-a|\nonumber \\ &+|\hat m_1-m_1|e^{-\eta_1}+(1+\varepsilon)|\hat m_2-m_2|e^{\eta_2} \nonumber\\
\le&
(1+\varepsilon)(|e_1|+|e_2|)+\varepsilon L_{a,\delta}|e_1|\nonumber \\ &+L_{m_1,\delta}E_{1,R}|e_1|+(1+\varepsilon)L_{m_2,\delta}E_{2,R}|e_2| \nonumber\\
\le&
L_{\mathrm{gain},R,\delta}(|e_1|+|e_2|),
\end{align}
for some constant $L_{\mathrm{gain},R,\delta}>0$. Finally, using
\begin{equation}
\Delta_u(\eta,e_1,e_2)=-e_2-(\hat a-a)+\beta\,\Delta_{\mathrm{gain}}(\eta,e_1,e_2),
\end{equation}
we obtain
\begin{eqnarray}
|\Delta_u(\eta,e_1,e_2)|
&\le&
|e_2|+|\hat a-a|+\beta|\Delta_{\mathrm{gain}}(\eta,e_1,e_2)| \nonumber\\
&\le&
|e_2|+L_{a,\delta}|e_1|+\beta L_{\mathrm{gain},R,\delta}(|e_1|+|e_2|) \nonumber\\
&\le&
C_{R,\delta}(|e_1|+|e_2|),
\end{eqnarray}
for all $\eta\in\mathcal T_R$, where $C_{R,\delta}:=\max\{1,L_{a,\delta}\}+\beta L_{\mathrm{gain},R,\delta}$. Since $\delta$ is fixed throughout the lemma, we write $C_R := C_{R,\delta}$, noting that $C_R$ depends on both $R$ and $\delta$.
\end{proof}

\section*{References}

\vspace*{-1.5\baselineskip}


\bibliography{sample}

\begin{thebibliography}{10}
\providecommand{\url}[1]{#1}
\csname url@samestyle\endcsname
\providecommand{\newblock}{\relax}
\providecommand{\bibinfo}[2]{#2}
\providecommand{\BIBentrySTDinterwordspacing}{\spaceskip=0pt\relax}
\providecommand{\BIBentryALTinterwordstretchfactor}{4}
\providecommand{\BIBentryALTinterwordspacing}{\spaceskip=\fontdimen2\font plus
\BIBentryALTinterwordstretchfactor\fontdimen3\font minus \fontdimen4\font\relax}
\providecommand{\BIBforeignlanguage}[2]{{%
\expandafter\ifx\csname l@#1\endcsname\relax
\typeout{** WARNING: IEEEtranS.bst: No hyphenation pattern has been}%
\typeout{** loaded for the language `#1'. Using the pattern for}%
\typeout{** the default language instead.}%
\else
\language=\csname l@#1\endcsname
\fi
#2}}
\providecommand{\BIBdecl}{\relax}
\BIBdecl

\bibitem{BARGO2026104629}
\BIBentryALTinterwordspacing
M.~Bargo and Y.~Simporé, ``Global stabilization and emergence tracking via aquatic control in an age-structured mosquito model,'' \emph{Nonlinear Analysis: Real World Applications}, vol.~92, p. 104629, 2026. [Online]. Available: \url{https://www.sciencedirect.com/science/article/pii/S1468121826000295}
\BIBentrySTDinterwordspacing

\bibitem{bhan2025stabilizationnonlinearsystemsunknown}
\BIBentryALTinterwordspacing
L.~Bhan, M.~Krstić, and Y.~Shi, ``Stabilization of nonlinear systems with unknown delays via delay-adaptive neural operator approximate predictors,'' 2025. [Online]. Available: \url{https://arxiv.org/abs/2509.26443}
\BIBentrySTDinterwordspacing

\bibitem{10374221}
L.~Bhan, Y.~Shi, and M.~Krstić, ``Neural operators for bypassing gain and control computations in {PDE} backstepping,'' \emph{IEEE Transactions on Automatic Control}, vol.~69, no.~8, pp. 5310--5325, 2024.

\bibitem{BHAN2025105968}
\BIBentryALTinterwordspacing
------, ``Adaptive control of reaction–diffusion {PDEs} via neural operator-approximated gain kernels,'' \emph{Systems \& Control Letters}, vol. 195, p. 105968, 2025. [Online]. Available: \url{https://www.sciencedirect.com/science/article/pii/S0167691124002561}
\BIBentrySTDinterwordspacing

\bibitem{chen1995Universal}
T.~Chen and H.~Chen, ``Universal approximation to nonlinear operators by neural networks with arbitrary activation functions and its application to dynamical systems,'' \emph{IEEE transactions on neural networks}, vol.~6, no.~4, pp. 911--917, 1995.

\bibitem{Chu2008-ix}
C.~Y.~C. Chu, H.-K. Chien, and R.~D. Lee, ``\BIBforeignlanguage{en}{Explaining the optimality of u-shaped age-specific mortality},'' \emph{\BIBforeignlanguage{en}{Theor. Popul. Biol.}}, vol.~73, no.~2, pp. 171--180, Mar. 2008.

\bibitem{Delbaere2020-gf}
I.~Delbaere, S.~Verbiest, and T.~Tyd{\'e}n, ``\BIBforeignlanguage{en}{Knowledge about the impact of age on fertility: a brief review},'' \emph{\BIBforeignlanguage{en}{Ups. J. Med. Sci.}}, vol. 125, no.~2, pp. 167--174, May 2020.

\bibitem{dochain2013automatic}
D.~Dochain, \emph{Automatic Control of Bioprocesses}.\hskip 1em plus 0.5em minus 0.4em\relax John Wiley \& Sons, 2013.

\bibitem{haacker2024stabilization}
P.-E. Haacker, I.~Karafyllis, M.~Krsti{\'c}, and M.~Diagne, ``Stabilization of age-structured chemostat hyperbolic {{PDE}} with actuator dynamics,'' \emph{International Journal of Robust and Nonlinear Control}, 2024.

\bibitem{iannelli2017basic}
M.~Iannelli and F.~Milner, ``The basic approach to age-structured population dynamics,'' \emph{Lecture Notes on Mathematical Modelling in the Life Sciences. Springer, Dordrecht}, 2017.

\bibitem{inaba2017age}
H.~Inaba, \emph{Age-Structured Population Dynamics in Demography and Epidemiology}.\hskip 1em plus 0.5em minus 0.4em\relax Springer, 2017, vol.~1.

\bibitem{karafyllis2017stability}
I.~Karafyllis and M.~Krstić, ``Stability of integral delay equations and stabilization of age-structured models,'' \emph{ESAIM: Control, Optimisation and Calculus of Variations}, vol.~23, no.~4, pp. 1667--1714, 2017.

\bibitem{karafyllis2025agestructuredchemostatsubstratedynamics}
\BIBentryALTinterwordspacing
I.~Karafyllis, D.~Theodosis, and M.~Krstić, ``The age-structured chemostat with substrate dynamics as a control system,'' 2025. [Online]. Available: \url{https://arxiv.org/abs/2511.09963}
\BIBentrySTDinterwordspacing

\bibitem{cdc2026version}
M.~Krsti{\'c}, I.~Karafyllis, L.~Bhan, and C.~Veil, ``Neural operators for control of age-structured population {PDEs},'' Honolulu, Hawaii, USA, Dec. 2026, submitted to 2026 IEEE 65th Conference on Decision and Control (CDC).

\bibitem{krstic2026neuraloperatorsagestructured}
------, ``Neural operators for control of age-structured population {PDEs},'' \url{https://github.com/lukebhan/NeuralOperatorsLoktaSharpePredatoryPrey}, 2026, github repository.

\bibitem{krstic2024Neural}
M.~Krstić, L.~Bhan, and Y.~Shi, ``Neural operators of backstepping controller and observer gain functions for reaction--diffusion {{PDEs}},'' \emph{Automatica}, vol. 164, p. 111649, Jun. 2024.

\bibitem{kurth2023control}
A.-C. Kurth and O.~Sawodny, ``Control of age-structured population dynamics with intraspecific competition in context of bioreactors,'' \emph{Automatica}, vol. 152, p. 110944, 2023.

\bibitem{kurth2021tracking}
A.-C. Kurth, K.~Schmidt, and O.~Sawodny, ``Tracking-control for age-structured population dynamics with self-competition governed by integro-{{PDEs}},'' \emph{Automatica}, vol. 133, p. 109850, 2021.

\bibitem{lamarque2025Adaptive}
M.~Lamarque, L.~Bhan, Y.~Shi, and M.~Krstić, ``Adaptive neural-operator backstepping control of a benchmark hyperbolic {{PDE}},'' \emph{Automatica}, vol. 177, p. 112329, 2025.

\bibitem{lamarque2025Gain}
M.~Lamarque, L.~Bhan, R.~Vazquez, and M.~Krstić, ``Gain scheduling with a neural operator for a transport {{PDE}} with nonlinear recirculation,'' \emph{IEEE Transactions on Automatic Control}, 2025.

\bibitem{Lanthaler2025Nonlocality}
\BIBentryALTinterwordspacing
S.~Lanthaler, Z.~Li, and A.~M. Stuart, ``Nonlocality and nonlinearity implies universality in operator learning,'' \emph{Constructive Approximation}, vol.~62, pp. 261--303, 2025. [Online]. Available: \url{https://doi.org/10.1007/s00365-025-09718-3}
\BIBentrySTDinterwordspacing

\bibitem{li2021Fourier}
Z.~Li, N.~Kovachki, K.~Azizzadenesheli, B.~Liu, K.~Bhattacharya, A.~Stuart, and A.~Anandkumar, ``Fourier {{Neural Operator}} for {{Parametric Partial Differential Equations}},'' May 2021.

\bibitem{loshchilov2018decoupled}
\BIBentryALTinterwordspacing
I.~Loshchilov and F.~Hutter, ``Decoupled weight decay regularization,'' in \emph{International Conference on Learning Representations}, 2019. [Online]. Available: \url{https://openreview.net/forum?id=Bkg6RiCqY7}
\BIBentrySTDinterwordspacing

\bibitem{lu2021Learning}
L.~Lu, P.~Jin, G.~Pang, Z.~Zhang, and G.~E. Karniadakis, ``Learning nonlinear operators via {{DeepONet}} based on the universal approximation theorem of operators,'' \emph{Nature machine intelligence}, vol.~3, no.~3, pp. 218--229, 2021.

\bibitem{LYU202513}
\BIBentryALTinterwordspacing
K.~Lyu, J.~Wang, Y.~Zhang, and H.~Yu, ``Neural operators for adaptive control of traffic flow models,'' \emph{IFAC-PapersOnLine}, vol.~59, no.~8, pp. 13--18, 2025, 5th IFAC Workshop on Control of Systems Governed by Partial Differential Equations - CPDE 2025. [Online]. Available: \url{https://www.sciencedirect.com/science/article/pii/S2405896325006433}
\BIBentrySTDinterwordspacing

\bibitem{mckendrick1925Applications}
A.~G. McKendrick, ``Applications of mathematics to medical problems,'' \emph{Proceedings of the Edinburgh Mathematical Society}, vol.~44, pp. 98--130, 1925.

\bibitem{schmidt2018yield}
K.~Schmidt, I.~Karafyllis, and M.~Krstić, ``Yield trajectory tracking for hyperbolic age-structured population systems,'' \emph{Automatica}, vol.~90, pp. 138--146, 2018.

\bibitem{sharpe1911problem}
F.~Sharpe and A.~Lotka, ``A problem in age-distribution,'' \emph{The London, Edinburgh, and Dublin Philosophical Magazine and Journal of Science}, vol.~21, no. 124, pp. 435--438, Apr. 1911.

\bibitem{veil2025stabilizationExtendedPreprint}
C.~Veil, M.~Krsti{\'c}, P.~McNamee, and O.~Sawodny, ``Stabilization of age-structured competing populations,'' \emph{arXiv preprint arXiv:2507.23013}, 2025.

\bibitem{veil2025stabilization}
C.~Veil, M.~Krstić, I.~Karafyllis, M.~Diagne, and O.~Sawodny, ``Stabilization of predator–prey age-structured hyperbolic {PDE} when harvesting both species is inevitable,'' \emph{IEEE Transactions on Automatic Control}, vol.~71, no.~1, pp. 123--138, 2026.

\bibitem{veil2025stabilizationagestructuredcompetingpopulations}
C.~Veil, P.~McNamee, M.~Krstić, and O.~Sawodny, ``Stabilization of age-structured competition (predator-predator) population dynamics,'' in \emph{2025 IEEE 64th Conference on Decision and Control (CDC)}, 2025, pp. 2058--2063.

\bibitem{wang2024Backstepping}
S.~Wang, M.~Diagne, and M.~Krsti{\'c}, ``Backstepping {{Neural Operators}} for 2x2 {{Hyperbolic PDEs}},'' Jul. 2024.

\bibitem{pmlr-v242-zhang24c}
\BIBentryALTinterwordspacing
Y.~Zhang, R.~Zhong, and H.~Yu, ``Neural operators for boundary stabilization of stop-and-go traffic,'' in \emph{Proceedings of the 6th Annual Learning for Dynamics \& Control Conference}, ser. Proceedings of Machine Learning Research, A.~Abate, M.~Cannon, K.~Margellos, and A.~Papachristodoulou, Eds., vol. 242.\hskip 1em plus 0.5em minus 0.4em\relax PMLR, 15--17 Jul 2024, pp. 554--565. [Online]. Available: \url{https://proceedings.mlr.press/v242/zhang24c.html}
\BIBentrySTDinterwordspacing

\end{thebibliography}
\bibliographystyle{IEEEtranS}

\balance
\vspace*{-1.5\baselineskip}

\begin{IEEEbiography}[{\includegraphics[width=1in,height=1.25in,clip,keepaspectratio]{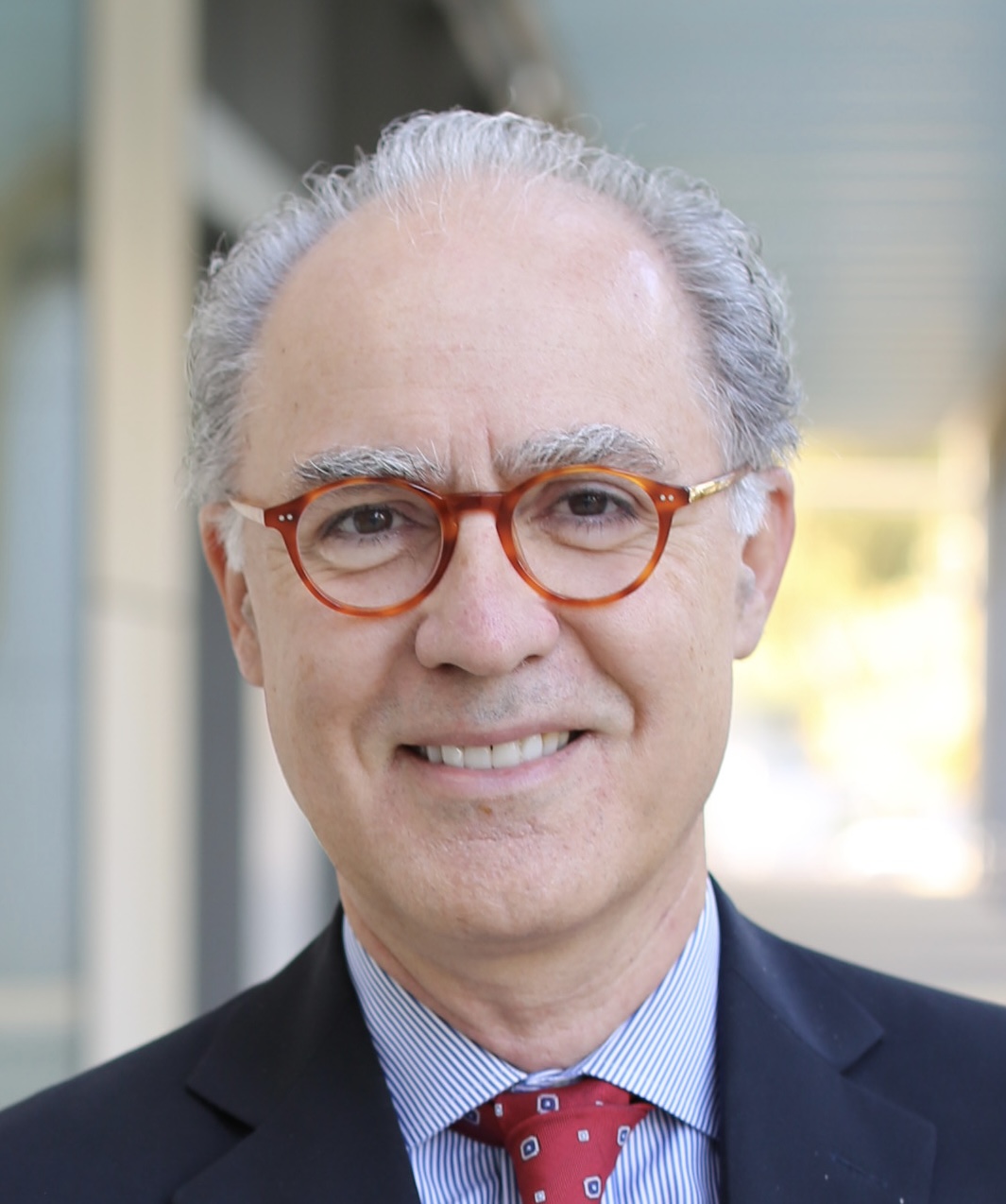}}]{Miroslav Krsti\'c} is a Distinguished Professor of Mechanical and Aerospace Engineering, holds the Alspach endowed chair, and is the founding director of the Center for Control Systems and Dynamics at UC San Diego. He also serves as Senior Associate Vice Chancellor for Research at UCSD. As a graduate student, Krstic won the UC Santa Barbara best dissertation award and student best paper awards at CDC and ACC. Krstic has been elected Fellow of IEEE, IFAC, ASME, SIAM, AAAS, IET (UK), and AIAA (Assoc. Fellow) - and as a foreign member of the Serbian Academy of Sciences and Arts and of the Academy of Engineering of Serbia. He has received the IEEE Roger W. Brockett Control Systems Award, Richard E. Bellman Control Heritage Award, Bode Lecture Prize, SIAM Reid Prize, ASME Oldenburger Medal, Nyquist Lecture Prize, Paynter Outstanding Investigator Award, Ragazzini Education Award, IFAC Nonlinear Control Systems Award, IFAC Ruth Curtain Distributed Parameter Systems Award, IFAC Adaptive and Learning Systems Award, IFAC Time-Delay Systems Lifetime Achievement Award, Chestnut textbook prize, AV Balakrishnan Award for the Mathematics of Systems, Control Systems Society Distinguished Member Award, the PECASE, NSF Career, and ONR Young Investigator awards, the Schuck (’96 and ’19) and Axelby paper prizes, and the first UCSD Research Award given to an engineer. Krstic is a Fellow-Ambassador of the French CNRS and has also been awarded the Miller Distinguished Visiting Professorship and Springer Visiting Professorship at UC Berkeley, the Distinguished Visiting Fellowship of the Royal Academy of Engineering, the Invitation Fellowship of the Japan Society for the Promotion of Science, and four honorary professorships outside of the United States. He serves as Editor-in-Chief of Systems \& Control Letters and has been serving as Senior Editor in Automatica and IEEE Transactions on Automatic Control, as editor of two Springer book series, and has served as Vice President for Technical Activities of the IEEE Control Systems Society and as chair of the IEEE CSS Fellow Committee. Krstic has coauthored nineteen books on adaptive, nonlinear, and stochastic control, extremum seeking, control of PDE systems including turbulent flows, and control of delay systems.
\end{IEEEbiography}

\begin{IEEEbiography}[{\includegraphics[width=1in,height=1.25in,clip,keepaspectratio]{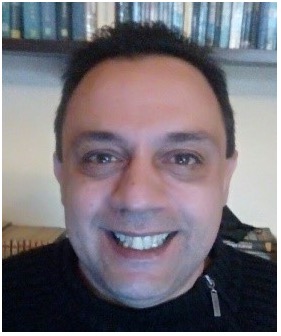}}]{Iasson Karafyllis} is a Professor in the Department of Mathematics, NTUA, Greece. He is a coauthor (with Z.-P. Jiang) of the book Stability and Stabilization of Nonlinear Systems, Springer-Verlag London, 2011 and a coauthor (with M. Krstic) of the books Predictor Feedback for Delay Systems: Implementations and Approximations, Birkhäuser, Boston 2017 and Input-to-State Stability for PDEs, Springer-Verlag London, 2019. Since 2013 he is an Associate Edi-tor for the International Journal of Control and for the IMA Journal of Mathematical Control and Information. Since 2019 he is an Associate Editor for Systems and Control Letters and Mathematics of Control, Signals and Systems. His research interests include mathematical control theory and nonlinear systems theory.
\end{IEEEbiography}

\begin{IEEEbiography}[{\includegraphics[width=1in]{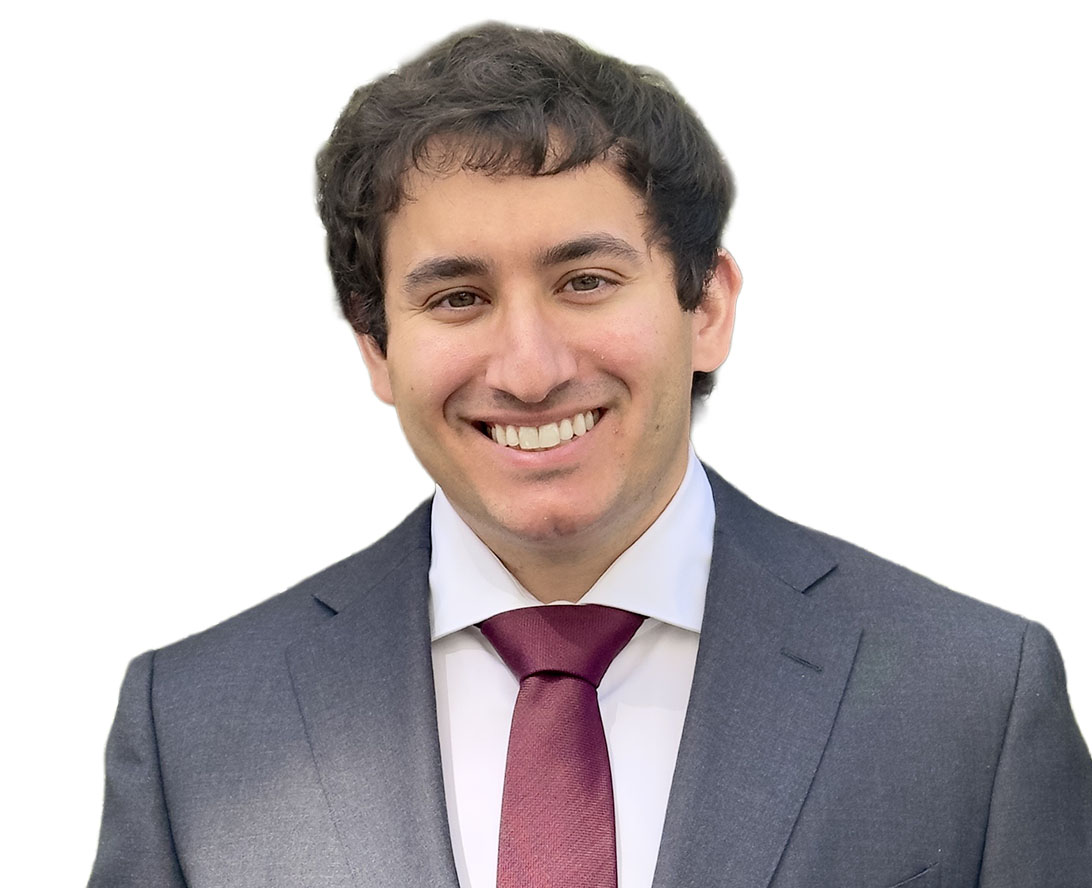}}]{Luke Bhan} received his B.S. and M.S. degrees in Computer Science, Math, and Physics from Vanderbilt University in 2022. He is currently pursuing his Ph.D. degree in Electrical and Computer Engineering
at the University of California, San Diego. His
research interests include neural operators, learning-based control, nonlinear delay systems, and partial differential equations. 
 \end{IEEEbiography}

\begin{IEEEbiography}[{\includegraphics[width=1in,height=1.25in,clip,keepaspectratio]{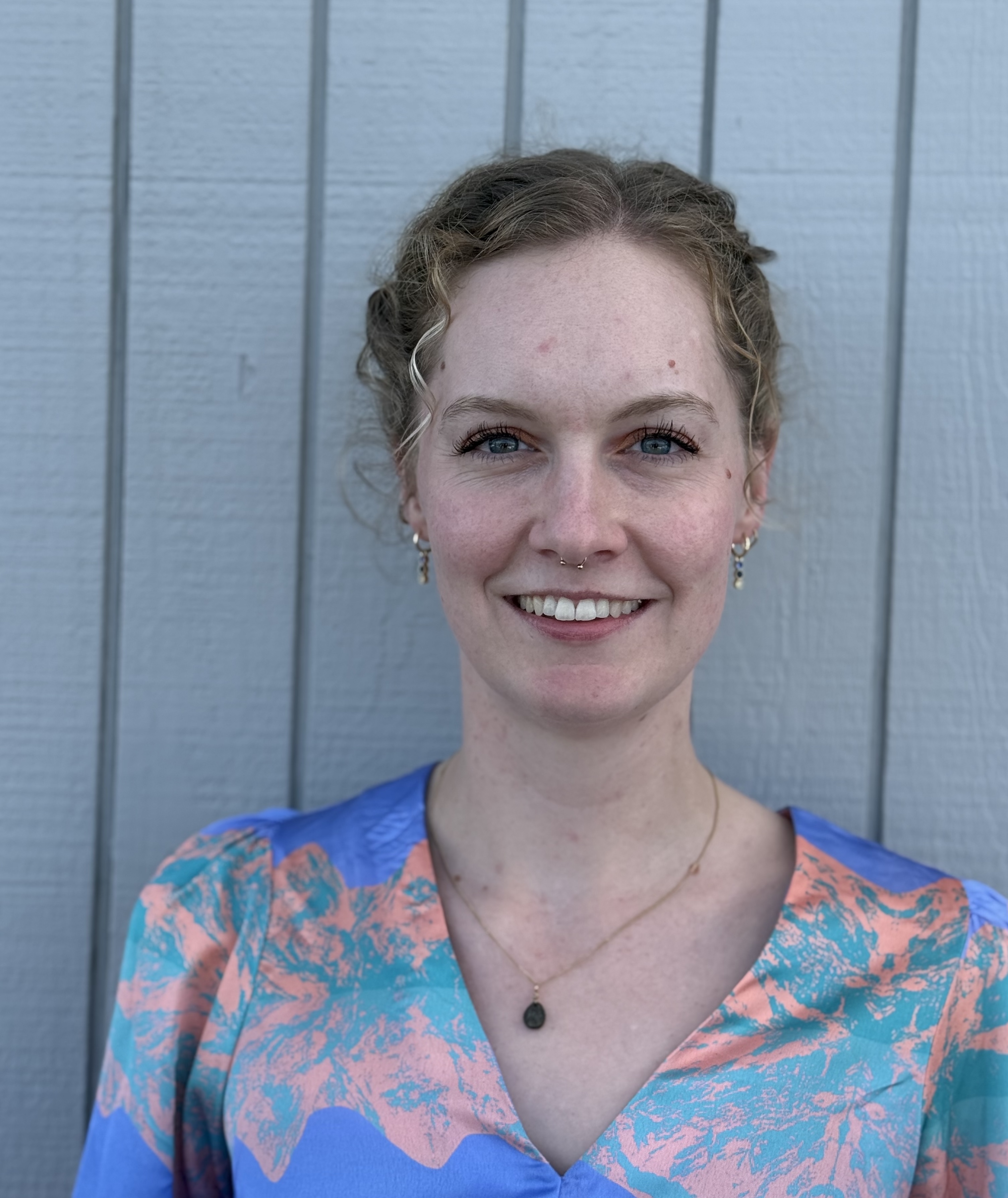}}]{Carina Veil} (Member, IEEE) is a postdoctoral researcher at KTH Royal Institute of Technology, Stockholm. Sweden. She received a B.Sc in biomedical engineering, M.Sc. in engineering cybernetics, and Ph.D in mechanical engineering from the University of Stuttgart, Germany, in 2017, 2020, and 2023, respectively. Prior to joining KTH, she has been a postdoctoral researchers at University of Stuttgart (2023-2025), Stanford University (2025-2026), as well as a visiting researcher at University of California San Diego in 2024.
Her research interests include controlling complex systems, with a special interest for health and sustainability applications, soft robots, and partial differential equations.

\end{IEEEbiography}

\end{document}